\documentclass[preprint]{imsart}

\RequirePackage{amsthm,amsmath,amsfonts,amssymb}
\RequirePackage[numbers]{natbib}

\RequirePackage[T1]{fontenc}
\RequirePackage{amsthm,amsmath}
\RequirePackage[numbers]{natbib}
\usepackage[textsize=footnotesize]{todonotes}
\usepackage{hyperref}
\hypersetup{colorlinks,citecolor=blue,urlcolor=blue}
\usepackage{cite}

\usepackage{graphicx}     
\usepackage{amsmath}
\usepackage{amssymb}
\usepackage{amsthm} 
\usepackage[countmax]{subfloat}
\usepackage{fourier}
\usepackage{xcolor} 
\usepackage{mathtools}
\usepackage{enumerate}
\usepackage{natbib}
\usepackage{url}
\usepackage{verbatim}
\usepackage{booktabs}
\usepackage{placeins}
\usepackage{multirow}
\usepackage{caption}
\usepackage{mathrsfs}
\usepackage{algorithm2e}
	\SetKwComment{Comment}{\#\#\#}{}
	\DontPrintSemicolon
	\SetAlgoLined
	\RestyleAlgo{boxruled}

\startlocaldefs
\theoremstyle{remark}

\newtheorem{thm}{Theorem}[section]
\newtheorem{lemma}{Lemma}[section]

\newtheorem{Remark}{Remark}[section]
\newtheorem{Corollary}{Corollary}[section]
\newtheorem{assumption}{Assumption}[section]

\def\R{\mathbb R}
\def\E{\mathbb E}

\def\cm{\text{Cum}}

\newcommand{\im}{\mathrm{i}}

\newcommand{\rnum}{\mathbb{R}}
\newcommand{\znum}{\mathbb{Z}}
\newcommand{\nnum}{\mathbb{N}}
\newcommand{\XT}[1]{X_{#1, T}}
\newcommand{\Xu}[1]{X^{(u)}_{#1}}

\newcommand{\cumbp}[2]{\kappa_{{#1};t_1,\ldots,t_{{#2}-1}}}

\newcommand{\F}{\mathcal{F}}

\allowdisplaybreaks

\newcommand\tageq{\addtocounter{equation}{1}\tag{\theequation}}

\endlocaldefs

\begin{document}

\begin{frontmatter}
\title{Adaptive Frequency Band Analysis for Functional Time Series}
%\title{A sample article title with some additional note\thanksref{T1}}
\runtitle{Frequency Band Learning for Functional Time Series}
\runauthor{P. Bagchi and S.A. Bruce}
%\thankstext{T1}{A sample of additional note to the title.}

\begin{aug}
%%%%%%%%%%%%%%%%%%%%%%%%%%%%%%%%%%%%%%%%%%%%%%
%%Only one address is permitted per author. %%
%%Only division, organization and e-mail is %%
%%included in the address.                  %%
%%Additional information can be included in %%
%%the Acknowledgments section if necessary. %%
%%%%%%%%%%%%%%%%%%%%%%%%%%%%%%%%%%%%%%%%%%%%%%
\author{\fnms{Pramita} \snm{Bagchi}\ead[label=e1]{pbagchi@gmu.edu}}
\and
\author{\fnms{Scott} \snm{Bruce}\ead[label=e2]{sbruce7@gmu.edu}}

%%%%%%%%%%%%%%%%%%%%%%%%%%%%%%%%%%%%%%%%%%%%%%
%% Addresses                                %%
%%%%%%%%%%%%%%%%%%%%%%%%%%%%%%%%%%%%%%%%%%%%%%
\address{Department of Statistics, George Mason University
% ,
% \printead{e1,e2}
}

\end{aug}

\begin{abstract}
The frequency-domain properties of nonstationary functional time series often contain valuable information. These properties are characterized through its time-varying power spectrum. Practitioners seeking low-dimensional summary measures of the power spectrum often partition frequencies into bands and create collapsed measures of power within bands. However, standard frequency bands have largely been developed through manual inspection of time series data and may not adequately summarize power spectra. In this article, we propose a framework for adaptive frequency band estimation of nonstationary functional time series that optimally summarizes the time-varying dynamics of the series. We develop a scan statistic and search algorithm to detect changes in the frequency domain. We establish theoretical properties of this framework and develop a computationally-efficient implementation. The validity of our method is also justified through numerous simulation studies and an application to analyzing electroencephalogram data in participants alternating between eyes open and eyes closed conditions.
\end{abstract}

\begin{keyword}[class=MSC2020]
\kwd[\href{https://mathscinet.ams.org/mathscinet/msc/msc2020.html}{MSC 2020 subject classifications. }]{62M15}
\kwd{62G20}
\end{keyword}

\begin{keyword}
\kwd{Frequency band estimation}
\kwd{functional time series}
\kwd{locally stationary}
\kwd{multitaper estimation}
\kwd{spectrum analysis.}
\end{keyword}

\end{frontmatter}
%%%%%%%%%%%%%%%%%%%%%%%%%%%%%%%%%%%%%%%%%%%%%%
%% Please use \tableofcontents for articles %%
%% with 50 pages and more                   %%
%%%%%%%%%%%%%%%%%%%%%%%%%%%%%%%%%%%%%%%%%%%%%%
%\tableofcontents

%%%%%%%%%%%%%%%%%%%%%%%%%%%%%%%%%%%%%%%%%%%%%%
%%%% Main text entry area:

\section{Introduction}

Functional  data  has  emerged  as  an  important  object  of interest in statistics in recent years as advances in technology have led to an abundance of high-dimensional and high resolution data.  While classical statistical methods often fail in this setting, functional data analysis techniques use the smooth structure of the observed process in order to model non-sparse, high-dimensional and high-resolution data.  The term “functional data analysis” was coined by \citep{ramsay1982data} and \citep{ramsay1991}, but the history of this area is much longer, dating back to \citep{grenander1950stochastic} and \citep{rao1958some}. The intrinsic high, or rather infinite, dimensionality of such data poses interesting challenges in both theory and computation and has garnered a considerable amount of attention within the statistics community.
%While functional data are now quite well studied, a relatively new field of research is emerging which aims to analyze the time-varying characteristics of such functional data. 
Functional time series data often arise in many important problems, such as analyzing forward curves derived from commodity futures \citep{HORVATH2020646}, daily patterns of geophysical and environmental data \citep{rubin2020}, demographic quantities, such as age-specific fertility or mortality rates studied over time \citep{shanghyndman2017}, and neurophysiological data, such as electroencephalography (EEG) and functional magnetic resonance imaging (fMRI), recorded at various locations in the brain \citep{GLENDINNING200779,STOEHR2020}. For example, NASA records surface temperatures for more than 5000 locations, and these readings are used to identify yearly temperature anomalies, which are crucial in studying global warming patterns \citep{nasa}.  In practice, such data are typically analyzed separately for each location, which is computationally expensive and fails to account for the spatial structure of the data.  However, it is reasonable to assume temperatures vary smoothly across locations, and thus ideal to analyze the collection of readings across locations as a functional data object. 

The frequency-domain properties of time series data, including the aforementioned functional time series data examples, often contain valuable information.  These properties are characterized through its power spectrum, which describes the
contribution to the variability of a time series from waveforms oscillating at different frequencies.  Practitioners seeking practical,
low-dimensional summary measures of the power spectrum often partition frequencies into bands and create collapsed measures of power within these bands. In practice, frequency band summary measures are used in a wide variety of contexts, such as to summarize seasonal patterns in environmental data, and to measure association between frequency-domain characteristics of EEG data and cognitive processes \citep{Klimesch1999}. In the scientific literature, standard frequency bands used for analysis have largely been developed through manual
inspection of time series data.  This is accomplished by noting prominent oscillatory patterns in the data and forming frequency bands
that largely account for these dominant waveforms.  For example, frequency-domain analysis of heart rate variability (HRV) data began in the late 1960s \citep{Billman2011} and led to the development of three primary frequency bands used to summarize power spectra: very low frequency (VLF) ($\le$ 0.04 Hz), low frequency (LF) (0.04-0.15 Hz), and high frequency (HF) (0.15-0.4 Hz) \citep{HRVtaskforce96}.  Within these frequency bands, collapsed measures of power are used to summarize the frequency-domain properties of the data and provide a basis for comparison.  However, fixed bands  that are not allowed to vary across subjects, covariates, experimental variables, or other
settings may not adequately summarize the power spectrum.  For example, \citep{Klimeschetal1998} shows that in the study of EEG data, different frequency bands within the fixed alpha band (8-13 Hz) reflect quite different cognitive processes.  This indicates that the \emph{number} of frequency bands needed to adequately summarize the power spectrum may depend on
experimental factors.  \citep{Doppelmayretal1998} also notes the \emph{endpoints} for the alpha frequency band may vary across subjects and proposes a data-adaptive definition for alpha band power. Both of these examples illustrate the need for a standardized, quantitative approach to frequency band estimation that provides a data-driven determination of the number of frequency bands and their respective endpoints.

For univariate nonstationary time series, a data-adaptive framework for identifying frequency bands that best preserve time-varying dynamics has been introduced in \citep{EBAsubmitted}. They develop a frequency-domain scan and hypothesis testing procedure and search algorithm to detect and validate changes in nonstationary behavior across frequencies.  
A rigorous framework for studying the frequency domain properties of functional time series was established in \citep{panaretos2013cramer} and \citep{panaretos2013fourier}. The spectral characteristics of locally stationary functional time series have been developed by \citep{vde16} and have been subsequently studied in several papers including \citep{Aue_Delft}, \citep{van2021similarity} and \citep{van2017nonparametric}.  In this paper, we develop a scan statistic for identifying the optimal frequency band structures for characterizing functional time series data. Developing appropriate scan statistics for the functional domain poses two important challenges. As the periodogram itself is an inconsistent estimator of the power spectrum, the univariate scan statistic was developed using a local multitapered periodogram estimator \citep{EBAsubmitted}. The corresponding test for detecting changes in the time-varying dynamics of the power spectrum across frequencies uses as asymptotic $\chi^2$ distribution approximation of the local multitapered periodograms.
However, the asymptotic behavior of the multitaper periodogram estimator for functional time series is not well studied.  In fact,  the approximation of such a  quantity  by  a $\chi^2$ distribution  is  not  possible  in  the  functional  case.   Hence,  a  functional generalization of this result
requires  a  completely different approach. In this paper, we derive a central limit theorem type result for the functional multitaper periodogram and establish a uniform distributional approximation of the functional scan statistics to a quadratic functional of a Gaussian process. The second challenge is implementation of the developed method, which is challenging largely due to the non-standard limit distribution and high-dimensionality of the data. In this work, we propose an efficient algorithm to detect the frequency band structure based on a finite dimensional projection of the functional scan statistic. Moreover, we propose a computationally-efficient and memory-smart modification of the algorithm based on the intrinsic structure of the asymptotic distribution.  These modifications enable the proposed methodology to be applied to longer and more densely-observed functional time series typically encountered in practice.  

The rest of the paper is organized as follows.  Section \ref{sec2} introduces the locally stationary functional time series model and offers a definition of the frequency-banded power spectrum.   Section \ref{sec3} first provides an overview of the local multitaper periodogram estimator for the power spectrum, which is used as a first step in the proposed procedure.   It should be noted that the proposed procedure can be carried out with any consistent time-varying spectral estimator, but this work focuses on the local multitaper periodogram estimator due to favorable empirical and theoretical properties established herein.  Section \ref{sec3} then details the components of the proposed analytical framework, including the scan statistics and their theoretical properties, an iterative algorithm for identifying multiple frequency bands, and another test statistic to determine if the power spectrum within a frequency band is stationary with respect to time.  Section \ref{sec:sec4} contains simulation results to evaluate empirical performance in estimating the number of frequency bands and their endpoints and provides an application to frequency band analysis of EEG data for an individual alternating between eyes open (EO) and eyes closed (EC) resting states.  Section \ref{sec:discussion} offers a discussion of the results and concluding remarks.  Proofs  for all theoretical results are provided in the Appendix of this paper.

\section{Nonstationary Functional Time Series in Frequency Domain} 
\def\theequation{2.\arabic{equation}}
\setcounter{equation}{0}
\label{sec2}
\subsection{Notation and Functional Set-up}
We begin by introducing some notation and results from functional analysis used in the paper. We typically assume observations are elements of a separable Hilbert space, $H$.  That is, $H$ is equipped with an inner product $\langle \cdot, \cdot \rangle$ and an associated norm $\|\cdot\|$ defined as $\|x \| = \langle x,x \rangle$ for $x \in H$. Let $\mathcal{L}$ be the space of bounded linear operators on $H$ with the norm
\begin{align*}
    \| \Phi \|_\mathcal{L} = \sup \{ \|\Phi(x)\| : \|x\| \leq 1 \}.
\end{align*}
An operator $\Phi \in \mathcal{L}$ is compact if there exists two orthonormal bases, $\{v_j\}$ and $\{f_j\}$, and a real sequence converging to zero, $\{\lambda_j\}$, such that for $x \in H$,
\begin{align*}
    \Phi(x) = \sum_{j=1}^\infty \lambda_j \langle x, v_j \rangle f_j.
\end{align*}
This representation is called the singular value decomposition. A compact operator with such a representation is said to be a Hilbert-Schmidt operator if $\sum_{j=1}^\infty \lambda_j^2 < \infty$. The space $\mathcal{S}$ of Hilbert-Schmidt operators is itself a separable Hilbert space with inner product
\begin{align*}
    \langle \Phi_1 , \Phi_2 \rangle_\mathcal{S} = \sum_{i=1}^\infty \langle \Phi_1(e_i),\Phi_2(e_i) \rangle
\end{align*}
where $\{ e_i\}$ is an arbitrary orthonormal basis whose selection does not alter the value of $\langle \Phi_1 , \Phi_2 \rangle_\mathcal{S}$. One can show that $\|\Phi\|_\mathcal{S}^2 = \sum_{j=1}^\infty \lambda_j^2$ and $\|\Phi\|_\mathcal{L} \leq \|\Phi\|_\mathcal{S}$. An operator $\Phi \in \mathcal{L}$ is symmetric if $ \langle \Phi(x),y \rangle = \langle x,\Phi(y) \rangle$ for all $x,y \in H$ and positive semi-definite if $\langle \Phi(x),x \rangle \geq 0$ for all $x \in H$.

In particular, we are interested in the Hilbert space $L^2([0,1]^k,\mathbb{C})$, the set of complex-valued measurable functions $x$ defined on $[0,1]^k$ satisfying $\int_{[0,1]^k} \vert x^2(u)\vert du < \infty$, equipped with the inner product
\begin{align*}
    \langle x,y\rangle  = \int_{[0,1]^k} x(t)\overline{y(t)} dt,
\end{align*}
where $\overline{y}$ denotes the complex conjugate of $y \in \mathbb{C}$.
Note for $x,y \in L^2([0,1]^k,\mathbb{C})$, we denote by $x=y$ the fact that $\|x-y\|= \int_{[0,1]^k} |x(t) - y(t)|^2 dt = 0$. The space of square integrable real valued functions $L^2([0,1]^k,\mathbb{R})$ is defined similarly.

An important class of integral operators on $L^2([0,1]^k,\mathbb{C})$ are those defined by
\begin{align*}
    \Phi[x](t) = \int_{[0,1]^k} \phi(t,s)x(s)ds
\end{align*}
for $x \in L^2([0,1]^k,\mathbb{C})$ with  kernel function $\phi \in L^2\left([0,1]^k\times[0,1]^k,\mathbb{C}\right)$. These operators are Hilbert-Schmidt if and only if $\int\int \vert \phi^2(s,t)\vert ds dt < \infty,$ in which case,
\begin{align*}
    \|\Phi\|_\mathcal{S}^2 = \int_{[0,1]^k} \int_{[0,1]^k} \vert \phi^2(s,t)\vert ds dt. 
\end{align*}
 If $\phi(s,t) = \phi(t,s)$ and $\int_0^1 \int_0^1 \phi(s,t)^2 ds dt \geq 0$, the integral operator $\Phi$ is symmetric and positive semi-definite.

\subsection{Power Spectrum of Locally Stationary Functional Time Series}
Let $\{X_h\}_{h \in \mathbb{Z}}$ be a weakly stationary functional time series  such that $X_t$ is a random element of $L^2([0,1],\R)$ for each $t\in\mathbb Z$  with expectation $\mu$ and auto-covariance kernel $r_h \in L([0,1]^2,\mathbb{R})$ at lag $h$ defined as
$$r_h(s,t) = \E\left((X_h(t)-\mu(t))(X_0(s)-\mu(s))\right),~~~~~~~~~t,s \in [0,1].$$
Let $\mathcal{R}_h$ be the corresponding autocovariance operator, i.e., the integral operator induced by the autocovariance kernel $r_h$.
The second order dynamics of this time series  can be completely described by the spectral density operator,  defined as the Fourier transform of $\{\mathcal{R}_h\}_{h\in\mathbb{Z}}$, acting on $L^2([0,1], \mathbb C)$, i.e., 
\[
\F_{\omega} = \sum_{h \in \mathbb{Z}} \mathcal{R}_h e^{-\im 2\pi \omega h} \quad \omega \in (0,0.5) \tageq \label{eq:statspo}.
\]
We will assume our data are centered and hence $\mu = 0$. When stationarity is violated, we can no longer define the spectral density in this manner for all time points. To meaningfully model nonstationarity, we consider
a triangular array $\{X_{t,T}:1\le t\le T\}_{T\in\mathbb N}$ as a doubly-indexed functional time series, where $X_{t,T}$ is a random element with values in $L^2([0,1],\R)$ for each $1\le t\le T$ and $T\in\mathbb N$. The processes $\{X_{t,T}:1\le t\le T\}$ are extended on $t \in \znum$ by setting $\XT{t} = \XT{1}$ for $t <1$ and $\XT{t} = \XT{T}$ for $t >T$. The sequence of stochastic processes $\{\XT{t}:t\in\mathbb Z\}$ indexed by $T\in\nnum$  is called {\it locally stationary} if for all rescaled times $u\in[0,1]$, there exists an $L^2([0,1],\R)$-valued strictly stationary process $\{X^{(u)}_t:t\in\mathbb Z\}$ such that
\begin{equation}\label{Local stationarity}
\Bigl\|\XT{t}-\Xu{t}\Bigr\|_{2}
\leq\big(\big|\tfrac{t}{T}-u\big|+\tfrac{1}{T}\big) \,P_{t,T}^{(u)}\qquad a.s.
\end{equation}
for all $1\leq t\leq T$, where $P_{t,T}^{(u)}$ is a positive real-valued process such that for some $\rho>0$ and $C<\infty$ the process satisfies $\E\big(\big|P_{t,T}^{(u)}\big|^\rho\big)<C$ for all $t$ and $T$ and uniformly in $u\in[0,1]$. If the second-order dynamics change gradually over time, the second order dynamics of the stochastic process $\{\XT{t}:t\in\mathbb Z\}_{T\in\mathbb N}$ are  completely described by the {\em time-varying spectral density kernel} given by
\begin{align}\label{eq:spdens}
f_{u,\omega}(\tau,\sigma) = \sum_{h\in \mathbb{Z}}
\E\left(X_{t+h}^{(u)}(\tau)X_t^{(u)}(\sigma)\right)
e^{-\im 2\pi \omega h},~~~~~~~ \omega \in (0,1/2).
\end{align}
The time varying spectral density operator $F_{u,\omega}$ is the operator induced by $f_{u,\omega}$ by right integration.
 Following the set-up of \citep{Aue_Delft}, in order to establish our asymptotic results, we impose the following technical assumptions on the time series under consideration. 

\begin{assumption}\label{cumglsp}
{Assume $\{\,X_{t,T} \colon t \in \znum\}_{T\in\nnum}$ is a locally stationary 
zero-mean stochastic process  and let $~\cumbp{k}{k}$ be a positive sequence in $L^2([0,1]^k,\rnum)$ independent of $T$ such that, for all $j=1,\ldots,k-1$ and some $\ell\in\nnum$,
\begin{align} 
\label{eq:kapmix}
\sum_{t_1,\ldots,t_{k-1} \in \mathbb{Z}} (1+|t_j|^\ell)\|\cumbp{k}{k}\|_2 <\infty.
\end{align}
Let us denote
\begin{equation}\label{eq:repstatap}
Y^{(T)}_{t}=\XT{t}-X_{t}^{(t/T)}
\qquad\text{and}\qquad
Y_{t}^{(u,v)}=\frac{\Xu{t}- X^{(v)}_{t}}{(u-v)}
\end{equation}
for $T\ge1$, $1\le t\le T$ and $u,v\in[0,1]$ such that $u\ne v$. Suppose furthermore that $k$-th order joint cumulants satisfy
\begin{enumerate}[(i)]\itemsep-.3ex
\item$\|\cm(\XT{t_1},\ldots,\XT{t_{k-1}},Y^{(T)}_{t_k}) \|_2 \le \frac{1}{T}\|\kappa_{k;t_1-t_k,\ldots,t_{k-1}-t_k}\|_2 $,
\item$\|\cm(X_{t_1}^{(u_1)},\ldots,X_{t_{k-1}}^{(u_{k-1})},Y_{t_{k}}^{(u_k,v)}) \|_2 \le \|\kappa_{k;t_1-t_k,\ldots,t_{k-1}-t_k}\|_2 $,
\item$\sup_u \|\cm(X_{t_1}^{(u)},\ldots,X_{t_{k-1}}^{(u)},X_{t_k}^{(u)}) \|_2 \le \|\kappa_{k;t_1-t_k,\ldots,t_{k-1}-t_k}\|_2$,
\item{$\sup_u \|\frac{\partial}{\partial u} \cm(X_{i,t_1}^{(u)},\ldots,X_{i,t_{k-1}}^{(u)},X_{i,t_k}^{(u)}) \|_2 \le \|\kappa_{k;t_1-t_k,\ldots,t_{k-1}-t_k}\|_2$.}
\end{enumerate}}
\end{assumption}
Under these assumptions, the spectral density operator $\mathcal{F}_{u,\omega}$ is a Hilbert-Schmidt operator and the spectral density kernel $f_{u,\omega} \in L^2\left([0,1]^2,\mathbb{C}\right)$ is twice-differentiable with respect to  $u$ and $\omega$.

We consider the demeaned power spectrum,
\begin{equation}
    g_{u,\omega}(\tau,\sigma) := f_{u,\omega}(\tau,\sigma) - \int_0^1 f_{u,\omega}(\tau,\sigma) du,
\end{equation}
and equivalently $\mathcal{G}_{u,\omega} = \mathcal{F}_{u,\omega} - \int_0^1 \mathcal{F}_{u,\omega} du.$
We assume $g$ admits a partition of the frequency space such that
\begin{equation}\label{eq:thdemps}
    g_{u,\omega}(\tau,\sigma) = \left\{
    \begin{array}{ll}
    g^{(1)}_{u}(\tau,\sigma) &~~\text{for} ~\omega \in (0,\omega_1),\\
    g^{(2)}_{u}(\tau,\sigma) &~~\text{for} ~\omega \in [\omega_1,\omega_2),\\
    ~~~~\vdots\\
    g^{(p)}_{u}(\tau,\sigma) &~~\text{for} ~\omega \in [\omega_{p-1},0.5).
    \end{array}
    \right.
\end{equation}
Equivalently, the spectral density operator $\mathcal{G}$ admits a partition
\begin{equation}
    \mathcal{G}_{u,\omega} = \left\{
    \begin{array}{ll}
    \mathcal{G}^{(1)}_{u} &~~\text{for} ~\omega \in (0,\omega_1),\\
    \mathcal{G}^{(2)}_{u} &~~\text{for} ~\omega \in [\omega_1,\omega_2),\\
    ~~~~\vdots\\
    \mathcal{G}^{(p)}_{u} &~~\text{for} ~\omega \in [\omega_{p-1},0.5).
    \end{array}
    \right.
\end{equation}
We want to estimate the number of segments in the partition $p$ and the associated partition points $(\omega_1,\omega_2,\dots,\omega_p)$ of the frequency space.

\section{Empirical Band Analysis} \label{sec3}
\def\theequation{3.\arabic{equation}}
\setcounter{equation}{0}
\subsection{Estimation of the Power Spectrum and Proposed Statistic}\label{sec3.1}
We start by approximating the time-varying  power spectrum by multitaper local periodograms based on the data. We consider $B$ equally-sized non-overlapping temporal blocks. For the sake of notational convenience, suppose that $T$ is a multiple of $B$, and let $T_B=T/B$ be the number of observations in each temporal block. 

For $b = 1,2,\dots,B$, $\omega\in(0,0.5)$, $k=1,2, \dots, K$ and $T\ge1$, the functional discrete Fourier transform (fDFT) is defined as a random function with values in $L^2([0,1],\mathbb C)$ given by
\begin{equation}\label{eq:fDFT}
	\widetilde{X}_T^{(k),b,\omega}
    :=\sum_{t=1}^{T}v_b^k(t) {I}_b(t/T)X_{t,T}e^{-\im 2\pi \omega t},
\end{equation}
where $v_b^k:\{1,2,\dots,T\}\mapsto \rnum$ is the $k$-th data taper for the $b$-th time segment.  We propose the use of sinusoidal tapers of the form
\begin{equation}
\label{eqn:sinetapers}
v_{b}^k(t) = \sqrt{\frac{2}{T_B+1}}\sin\frac{\pi k [t-(b-1)T_B]}{T_B+1}, \; \mathrm{for} \; k=1,\ldots,K,
\end{equation}
which are orthogonal for $t \in [(b-1)T_B+1,\ldots,bT_B]$.  Sinusoidal tapers are more computationally efficient than the Slepian tapers proposed in \citep{Thomson1982}, which require numerical eigenvalue decomposition to construct the tapers, and can achieve similar spectral concentration with significantly less local bias \citep{RidelandSidorenko1995}.  Another advantage is that the bandwidth, $bw$, which is the minimum separation in frequency between approximately uncorrelated spectral estimates, can be fully determined by setting the number of tapers, $K$, appropriately \citep{Waldenetal1995}.  More specifically, 
\begin{equation}
\label{eqn:bwsinetapers}
bw = \frac{K+1}{T_B+1}.
\end{equation}
The k-th single local periodogram kernel is then defined by
$$\widehat{f}_{b,\omega}^{(k)}(\tau,\sigma) = \widetilde{X}_T^{(k),b,\omega}(\tau)\widetilde{X}_T^{(k),b,-\omega}(\sigma)= \widetilde{X}_T^{(k),b,\omega}(\tau)\overline{\widetilde{X}_T^{(k),b,\omega}(\sigma)}.$$

The multitaper estimator of the local power spectrum within the $b$-th block is then defined to be the average of all $K$ single taper estimators
\begin{eqnarray}
        \widehat{f}_{b,\omega}^{(mt)}(\tau,\sigma) = \frac{1}{K}\sum_{k=1}^K \widehat{f}_{b,\omega}^{(k)}(\tau,\sigma).
\end{eqnarray}
The final estimator of the time-varying power spectrum is then given by
\begin{eqnarray}\label{eq:multpdg}
    \widehat{f}_{u,\omega}^{(mt)}(\tau,\sigma) = \sum_{b=1}^B I_b(u) \widehat{f}_{b,\omega}^{(mt)}(\tau,\sigma).
\end{eqnarray}
In practice, this estimator requires appropriate selection of two tuning parameters: $B$ and $K$.  The number of time blocks, $B$, balances frequency and temporal properties.   It should be selected small enough to ensure sufficient frequency resolution and so the central limit theorem holds for the local tapered periodgrams.  It should be selected large enough so that data within each block are approximately stationary, which depends on the signal under study.  
% In many applications, scientific guidelines exist for the selection of $B$.  For example, it is recommended to use time blocks 2-5 minutes in length for analyzing HRV time series \citep{HRVtaskforce96}. 
Asymptotic results presented in the next section indicate the optimal rate for $B$ is $T^{1/2}$.   In the absence of scientific guidelines, we recommend selecting $B$ as the factor of $T$ closest to $T^{1/2}$.   The number of tapers, $K$, controls the smoothness of local spectral estimates.   Asymptotic results presented in the next section indicate an appropriate choice would be $K=\max \left(1,\left \lfloor \min(T_B,B)^{1/2}\right \rfloor\right)$.

% General guidance for multitaper estimation is to select $K$ such that the product of the number of Fourier frequencies and the bandwidth, which is a function of $K$, is equal to some predefined value, the popular default of which is $4$ \citep{PercivalWalden1993}.  In our setting, this leads to selecting $K = 8 \left(T + B\right)/T - 1$.

We further define the estimated demeaned time-varying spectra as
\begin{equation}\label{eq:demps}
    \widehat{g}_{b/B,\omega}(\tau,\sigma) = \widehat{f}_{b/B,\omega}^{(mt)}(\tau,\sigma) - \frac{1}{B}\sum_{l=1}^B \widehat{f}_{l/B,\omega}^{(mt)}(\tau,\sigma).
\end{equation}
Let $\widetilde{g}_{u,\omega_0,\delta}$ be the average of demeaned time-varying spectrum for $\omega \in [\omega_0,\omega_0+\delta).$
In order to identify frequency partition points, we consider an integrated scan statistic defined as
% \begin{align}
%     Q_{\omega_0,\delta}(\tau,\sigma) &~=~  \displaystyle\sum_{b=1}^B  \left(\widehat{g}_{b/B,\omega_0+\delta}(\tau,\sigma)- \widetilde{g}_{b/B,\omega_0,\delta}(\tau,\sigma)\right)^2
% \end{align}
% and an integrated version of the previous statistic
\begin{align}
    Q_{\omega_0,\delta} &~=~  \displaystyle\sum_{b=1}^B \int_0^1 \int_0^1 \left(\widehat{g}_{b/B,\omega_0+\delta}(\tau,\sigma)- \widetilde{g}_{b/B,\omega_0,\delta}(\tau,\sigma)\right)^2d\tau d\sigma.
\end{align}

\subsection{Asymptotic Properties of the Scan Statistic}
The basic intuition behind this scan statistic is that the estimator $\widehat{g}$ defined in \eqref{eq:demps} asymptotically behaves like the theoretical demeaned power spectrum defined in \eqref{eq:thdemps}. The scan statistic compares the value of $\widehat{g}$ at $\omega_0 + \delta$ with the average of the same quantity over the interval $(\omega_0,\omega_0 + \delta)$. Hence, if one of the partition points $\omega_k$ defined in \eqref{eq:thdemps} is present in the interval $(\omega_0,\omega_0 + \delta)$, the scan statistic should take a large value, and  it should be bounded and close to zero otherwise. For rest of the paper, we will consider the following asymptotic scheme.

\begin{assumption}
\label{assymp}
Assume $T \to \infty$, $B \to \infty$, $K \to \infty$, $T/K^2 \to \infty$ and $T/BK \to \infty,$
\end{assumption}

Our first result characterizes the asymptotic behavior of the estimated demeaned power spectrum.

\begin{lemma}\label{lem:ghat}
Under \autoref{cumglsp} and \autoref{assymp}, for every  $b, b_1, b_2\in \{1,2,\dots,B\}, \omega \in (0,0.5)$, and $\tau , \tau_1, \tau_2, \sigma, \sigma_1, \sigma_2 \in [0,1]$,
\begin{align}\label{ghat}
    &\mathbb{E}\left(\widehat{g}_{b/B,\omega}(\tau,\sigma)\right) = g_{u_b,\omega}(\tau,\sigma) + O\left(\log(T_B)/T_B \right) + \log \left(1/T \right) \nonumber\\
    &\text{Cov}\left(\widehat{g}_{b/B,\omega_1}(\tau_1,\sigma_1),\widehat{g}_{b/B,\omega_2}(\tau_2,\sigma_2) \right)=  \frac{(1-2/B)}{K}F(u_b,\omega_1,\omega_2,\tau_1,\sigma_1,\tau_2,\sigma_2)\nonumber\\
    & ~~ \hspace{2.25 in}+ \frac{1}{B^2K}\sum_{l=1}^B F(u_l,\omega_1,\omega_2,\tau_1,\sigma_1,\tau_2,\sigma_2) + o(1).\nonumber\\
    &\text{Cov}\left(\widehat{g}_{b_1/B,\omega_1}(\tau_1,\sigma_1),\widehat{g}_{b_2/B,\omega_2}(\tau_2,\sigma_2) \right)=  -\frac{1}{BK}F(u_{b_1},\omega_1,\omega_2,\tau_1,\sigma_1,\tau_2,\sigma_2)\nonumber\\
   &~~~~~~~~~~~~~~~~~~~~ - \frac{1}{BK} F(u_{b_2},\omega_1,\omega_2,\tau_1,\sigma_1,\tau_2,\sigma_2)\nonumber + \frac{1}{B^2K}\sum_{l=1}^B F(u_{b_l},\omega_1,\omega_2,\tau_1,\sigma_1,\tau_2,\sigma_2) + o(1),
\end{align}
where
$$F(u,\omega_1,\omega_2,\tau_1,\sigma_1,\tau_2,\sigma_2) := f_{u,\omega_1}(\tau_1,\tau_2)f_{u,\omega_2}(\sigma_1,\sigma_2) + f_{u,\omega_2}(\tau_1,\sigma_2)f_{u,\omega_2}(\tau_2,\sigma_1),$$
% $T_B = T/B$, the number of observations in each block, 
and $u_b$ is the mid-point of the $b$-th block. 
\end{lemma}

\begin{Remark}
Note that the assumptions $T/BK \to \infty$ and $K \to \infty$ together guarantee $T_B \to \infty$. Therefore \autoref{lem:ghat} implies that $\mathbb{E}\left(\widehat{g}_{b/B,\omega}(\tau,\sigma)\right) \to g_{u_b,\omega}(\tau,\sigma)$ under our asymptotic scheme. Moreover, the covariance kernel of the process $\widehat{g}$ is $O(1/K)$. This guarantees for fixed $b$ and $\omega$, $\widehat{g}_{b/B,\omega}$ consistently estimates $g_{u_b,\omega}$. A stronger uniform result on $b$ and $\omega$ can be established under additional moment assumptions. Additionally, the covariance across different time blocks is of the order $O(1/BK)$, which converges to $0$ faster than the covariance kernel within the same block.
\end{Remark}

Next, we investigate the asymptotic behavior of the scan statistic itself. The next two theorems describe the asymptotic behavior of the scan statistic under the absence and presence of a jump point in the interval $(\omega_0,\omega_0+\delta)$ respectively.

\begin{thm}
\label{thm:null}
Assume $g_{u,\omega}(\tau,\sigma) = g^{(0)}_{u}(\tau,\sigma)$ for all $\tau,\sigma \in [0,1]$ and $\omega \in [\omega_0,\omega_0+\delta].$ Under \autoref{cumglsp} and \autoref{assymp},
% $$ Q_{\omega_0,\delta}(\tau,\sigma) \stackrel{d}{=} \frac{1}{K}\sum_{b=1}^B \left(\mathcal{G}_b^2(\tau,\sigma) + o_p(1)\right),$$ 
% and  
$$Q_{\omega_0,\delta}  \stackrel{d}{=}\frac{1}{K}\sum_{b=1}^B \left( \| \mathcal{G}_b\|_2^2 + o_p(1)\right),$$  
where $\{\mathcal{G}_b\}$ is a collection of zero-mean Gaussian processes in   $L^2 \left([0,1]^2\right)$ with  covariance structure given in \eqref{var:diagonal} and \eqref{var:offdiag}.
\end{thm}

By Lemma \ref{lem:gb}, the quantity $\frac{1}{B}\sum_b\| \mathcal{G}_b\|^2 = O_p(1)$ as $B \to \infty$. Therefore, \autoref{thm:null} suggests that the quantity $K/B Q_{\omega_0,\delta}$ is asymptotically bounded. Hence, as long as the last quantity diverges under our asymptotic set-up, we can construct a consistent test for the existence of a partition point in frequency domain. The next theorem guarantees that is indeed the case.

\begin{thm}
\label{thm:alt}
Assume that
\begin{align*}
    g_{u,\omega}(\tau,\sigma) = \left\{\begin{array}{ll}
       g^{(1)}_{u}(\tau,\sigma)  & \text{for } \omega \in [\omega_0,\omega^*) \\
       g^{(2)}_{u}(\tau,\sigma)  & \text{for } \omega \in [\omega^*,\omega_0+\delta].
    \end{array}\right.
    \end{align*}
  Under \autoref{cumglsp} and \autoref{assymp},
% $$ Q_{\omega_0,\delta}(\tau,\sigma) \stackrel{d}{=} \frac{1}{K}\sum_{b=1}^B \left(\mathcal{G}_b^2(\tau,\sigma) + o_p(1)\right) + B\left(\frac{\omega^* - \omega_0}{\delta}\right)^2\int_{0}^1\left(g^{(1)}_{u}(\tau,\sigma) - g^{(2)}_{u}(\tau,\sigma)\right)^2du + o_p(1),$$
$$Q_{\omega_0,\delta} \stackrel{d}{=}\frac{1}{K}\sum_{b=1}^B \left( \| \mathcal{G}_b\|_2^2 + o_p(1)\right) + B\left(\frac{\omega^* - \omega_0}{\delta}\right)^2\int_{0}^1\int_0^1 \int_0^1 \left(g^{(1)}_{u}(\tau,\sigma) - g^{(2)}_{u}(\tau,\sigma)\right)^2 d\tau d\sigma du + O_p(B/K),$$
where $\{\mathcal{G}_b\}$ are as defined in Theorem \ref{thm:null}.
\end{thm}

\begin{Remark}
If there is no significant frequency partition point in $[\omega_0,\omega_0+\delta]$, we have
$$Q_{\omega_0,\delta} = O(B/K),$$
and in the presence of a change,
$$Q_{\omega_0,\delta} = O(B/K) + O(B).$$
Therefore, we expect to see a large spike in the scan statistic around the frequency where the dynamics of the power spectrum changes.
\end{Remark}

\begin{Remark}
Suppose we are interested in testing the null hypothesis that the kernel $g_{u,\omega}$ does not change in $\omega$ on the interval  $(\omega_0,\omega_0+\delta)$. Let $c_{\alpha}$ be the $(1-\alpha)$-th quantile of  $\frac{1}{B} \sum_{b=1}^B  \| \mathcal{G}_b\|_2^2 $ for large $B$. Then an asymptotic level $\alpha$ test for this hypothesis can be constructed by rejecting the null when $K Q_{\omega_0,\delta} > B c_{\alpha}$. In fact, the next Corollary suggests that this is a consistent test.
\end{Remark}
\begin{Corollary}
Under the assumptions of Theorem \ref{thm:alt}, we have $$P\left( Q_{\omega_0,\delta} > \frac{1}{K} \sum_{b=1}^B  \VERT \mathcal{G}_b\VERT_2^2\right) \to 1.$$
\end{Corollary}

\subsection{An Iterative Algorithm}
The results discussed in the previous section provide a consistent test to find a single partition point in the frequency space. In this section, we extend this framework to detect multiple frequency partition points using an iterative search algorithm that uses this scan statistic to efficiently explore the frequency space and identify all frequency partition points. 

For computational efficiency, we search across frequencies in batches of size $n_{\max}$.  This limits the number of calculations required to approximate the null distribution of the scan statistic to avoid undue computational burden.  In particular,  $\omega_0$ is first fixed at a value near $0$, and we conduct a test for partition points in the interval $(\omega_0, \omega_0 +\delta]$ comparing the statistic $Q_{\omega_0,\delta}$ against the null distribution described in Theorem \ref{thm:null}  for different values of $\delta < n_{\max}/T_B$. A Hochberg step-up procedure \citep{hochberg1988} is then used to test for the presence of a partition in the frequency domain at each particular frequency $\omega_0+\delta$.  If the procedure returns a set of frequencies for which the null hypothesis (no partition point) is rejected, we then add the smallest frequency in this set to the estimated frequency band partition, increase $\omega_0$ to be just larger than this newly found frequency partition point, and repeat the process.  If the procedure does not return any frequencies for which the null hypothesis is rejected, we increase $\omega_0$ to be just larger than the largest frequency tested in the current batch and repeat the process.  The procedure continues until all frequencies have been evaluated as potential partition points.  To better visualize this procedure, consider that the traversal of this procedure across frequencies resembles the movement of an inchworm.  A complete algorithmic representation of the search procedure in its entirety is available in Algorithm \ref{algo:eba}. 

P-values are obtained by comparing observed test statistics with a simulated distribution of the limiting random variable described in \autoref{thm:null}. To simulate from the null distribution, we generate $d_0$ draws from the collection of zero-mean Gaussian processes $\{\mathcal{G}_b\}$  with covariance structure given in \eqref{var:diagonal} and \eqref{var:offdiag}.  The covariance is estimated using the estimates of the demeaned time-varying power spectrum based on the local multitaper power spectrum estimator $\hat{f}^{(mt)}_{u,\omega}$ described in Section \ref{sec3.1}.

\begin{algorithm}[ht!]

  %inputs
  \SetKwInOut{Input}{Input}
  \Input{Demeaned time-varying multitaper power spectrum estimates, $\hat{g}_{b/B,\omega_k}(\tau,\sigma)$, for $b=1,\ldots,B$ and $\omega_{k}=k/T_B$ for $k=1,\ldots,N_B=\lfloor T_B/2 \rfloor - 1$ \newline
  Number of tapers $K$, significance level $\alpha$,
  number of frequencies tested in each pass $n_{\max}$,
  number of draws for approximating $p$-values, $d_0$}
  
  %outputs
  \SetKwInOut{Output}{Output}
  \Output{Estimated number of partition points,  $\widehat{p}$ \newline
  Estimated partition points,  $\widehat{\boldsymbol{\omega}}_{\widehat{p}} =
    \{\widehat{\omega}_1,\widehat{\omega}_2,\ldots,\widehat{\omega}_{\widehat{p}-1}\}$}
    
  %initialization 
  $\widehat{p} \gets 1$, $\widehat{\boldsymbol{\omega}}_{\widehat{p}} \gets \{\}$, $\mathrm{stop} \gets 0$, $\epsilon \gets \frac{K+1}{T_B+1}$,  $k^* \gets \left \lceil T_B \epsilon \right \rceil$, $\omega_0 \gets \omega_{k^*}$\;
  %loop
  \While{$\mathrm{stop} \neq 1$}{

%   \uIf{$\omega_{0}> \omega_{N_B}-2 \epsilon$}{
%   \vspace{1mm}
%   stop $\gets 1$
  
%   }
%     \Else{
    $k_{\min} \gets \left \lceil T_B \epsilon \right \rceil$,
   $k_{\max} \gets \min \left(k_{\min}+n_{\max}-1, \left \lfloor N_B - k^{*}-T_B \epsilon \right \rfloor\right)$\;
   
  \vspace{1mm}
  
  Compute test statistics $Q_{\omega_0,\delta_k} \; \forall \; \delta_k = k/T_B, k \in \{k_{\min},k_{\min}+1,\ldots,k_{\max}\}.$\;
   
    \vspace{1mm}
    
Simulate $d_0$ draws, $\mathbf{Q}^{H_0}_{\omega_0,\delta_k}=\left\{Q^{H_0,i}_{\omega_0,\delta_k}\right\}_{i=1}^{d_0}$, from the limiting null distribution of $Q_{\omega_0,\delta_k}$. 
% described in \autoref{thm:null}.
\;

  \vspace{1mm}
 Approximate $p$-values, $\hat{p}(k) = \frac{1}{d_0}\sum_{i=1}^{d_0} I \left(Q^{H_0,i}_{\omega_0,\delta_k} > Q_{\omega_0,\delta_k}\right)$
to test $H_{0}(k):g_{u,\omega}(\tau,\sigma)=g_u^{(0)}(\tau,\sigma)$ for $\omega \in [\omega_0,\omega_0+\delta_k]$.  \;

%   Approximate $p$-values, $p(k)$, to test $H_{0}(k):g_{u,\omega}(\tau,\sigma)=g_u^{(0)}(\tau,\sigma)$ for $\omega \in [\omega_0,\omega_k]$ for $ k \in \{k_{\min},k_{\min}+1,\ldots,k_{\max}\}$ by comparing $Q_{\omega_0,\delta_k}$ with a simulated distribution of the limiting random variable described in \autoref{thm:null}.\;
  
   \vspace{1mm}
  
    % Identify $\omega_k$ for which $H_{0}(k)$ are rejected using Hochberg step up procedure controlling FWER at level $\alpha$,
    Identify $\mathbf{R_{\alpha}} = \{\omega_0+\delta_k:H_{0}(k) \; \mathrm{rejected}\}$ using level-$\alpha$ Hochberg step-up procedure.\;

\vspace{1mm}

    \uIf{$\mathbf{R_{\alpha}} = \{\}$ }{ $k^{*} \gets T_B (\omega_0 + \delta_{k_{\max}}), \; \omega_0 \gets k^{*}/N_B$\;}
    \Else{
    $\hat{\omega}^* \gets \min \mathbf{R_{\alpha}}, \; \widehat{\boldsymbol{\omega}}_{\widehat{p}} \gets \widehat{\boldsymbol{\omega}}_{\widehat{p}} \cup \{ \hat{\omega}^*\}, \; \widehat{p} \gets \widehat{p} + 1\;$
    
   $k^{*} \gets \left \lceil T_B (\hat{\omega}^* + \epsilon) \right \rceil, \; \omega_0 \gets k^{*}/N_B$\;} 

\lIf{$\omega_{0}> \omega_{N_B}-2 \epsilon$}{
  \vspace{1mm}
  stop $\gets 1$
}  
  
    }
\Return{$\widehat{p}, \widehat{\boldsymbol{\omega}}_{\widehat{p}}$}\;
\caption{{\sc Inchworm Frequency Band Search Algorithm}}
\label{algo:eba}    
\end{algorithm}

\noindent\begin{Remark}\label{rmk:comp}
Computation cost is a critical aspect for analyzing functional time series data. In order to make our algorithm more efficient, we use two particular adjustments
\begin{enumerate}[(i)]
    \item \textit{A block diagonal approximation of the asymptotic covariance structure:} While approximating the asymptotic quantiles by simulating $\{\mathcal{G}_b\}$, we ignore the covariance across different time blocks $b$ given by \eqref{var:offdiag}. This allows for simulating from a collection of $B$ independent, lower dimensional Gaussian processes, which can be carried out in parallel for computational efficiency.  Note that the covariance kernel for $\mathcal{G}_b$ given in \eqref{var:diagonal} is $O(1)$ and the covariance kernel in \eqref{var:offdiag} is $O(1/B)$. Therefore, for large $B$, the approximation works reasonably well.
    
    \item \textit{Choice of the tuning parameter $n_{\max}$:} Notice that number of calculations required to estimate the covariance kernel given in \eqref{var:offdiag} grows with the number of frequencies between $\omega_0$ and $\omega_0+\delta$.  By traversing frequencies by testing in smaller batches, this reduces computation times significantly compared to letting $\delta$ vary over the whole frequency domain in a single pass.  
\end{enumerate}

\end{Remark}

All computations in what follows were performed using \texttt{R} 4.0.3 \citep{rcite}.  With the computational improvements introduced above, the run time using the proposed search algorithm for a single realization of functional white noise of length $T=2000$ with $R=5$ observed points in the functional domain using $B=5$ time blocks and $K=15$ tapers with $n_{\max}=40$ and $d_0=100000$ is 5 to 6 minutes on a laptop with 32GB RAM and a 8-core 2.4 GHz processor on a Mac operating system.  Without these improvements, the run time for the same data settings is over 10 hours, and many of the larger data settings considered in Section \ref{sec:sims} are not computationally feasible.  

% \subsection{Estimation of the Asymptotic Variance and Choice of the Tuning Parameters}

% Something about the block diagonal structure!!

\subsection{Test for Stationarity Within Frequency Blocks}

The proposed methodology produces homogeneous regions of frequencies in which the power spectrum varies only across time. It is natural to then seek to identify frequency bands for which the second order structure is stationary. Specifically, within a frequency band $[\omega_1,\omega_2]$, the power spectrum $f_{u,\omega}(\tau,\sigma) = f_u(\tau,\sigma)$ for all $u, \tau,\sigma \in [0,1]$ and $\omega \in [\omega_1,\omega_2].$ Furthermore, if the process in stationary within this region, the power spectrum is constant across time $u$, i.e., $f_{u,\omega}(\tau,\sigma) = f(\tau,\sigma)$ within than band. In this situation, the demeaned power spectrum
$$g_{u,\omega}(\tau,\sigma) = f_{u,\omega}(\tau,\sigma) - \int_0^1 f_{u,\omega}(\tau,\sigma) du \equiv 0,$$
for for all  $u, \tau,\sigma \in [0,1]$ and $\omega \in [\omega_1,\omega_2].$ Therefore, to test stationarity within a frequency band, we consider null hypothesis
$$H_0: g_{u,\omega}(\tau,\sigma) \equiv 0, ~~ \text{ almost everywhere  } u,\tau,\sigma \in [0,1],~ \omega \in [\omega_1,\omega_2],$$
against the alternative
$$H_a: g_{u,\omega}(\tau,\sigma) \neq 0, ~~ \text{on a set of positive Lebesgue measure.}$$
To test this hypothesis, we propose the test statistic
$$Q^0(\omega_1,\omega_2) = \frac{1}{B} \sum_{b=1}^B\int_{\omega_1}^{\omega_2}\int_0^1\int_0^1 \left\vert \widehat{g}_{b/B,\omega}(\tau,\sigma) \right\vert^2d\tau d\sigma d\omega.$$
As $\widehat{g}$ is a consistent estimator of the true demeaned power spectrum $g$, this  statistic is close to $0$ under $H_0$ and takes large positive values under the alternative. Now we formalize this idea through the asymptotic distribution of this statistic.

\begin{thm}\label{thm:sttest}
Assume that $g_{u,\omega}(\tau,\sigma) = g_{u}(\tau,\sigma) = 0$ for all $u, \tau,\sigma \in [0,1]$ and $\omega \in [\omega_1,\omega_2].$ Under \autoref{cumglsp} and \autoref{assymp}, 
\begin{equation}\label{stat_st}
    Q^{0}(\omega_1,\omega_2) \stackrel{d}{=} \frac{1}{K}\sum_{b=1}^B \left(\| \mathcal{H}_b\|_2^2+ o(1)\right), 
\end{equation}
where $\{\mathcal{H}_b\}$ is a collection of zero-mean Gaussian Process in $L^2[0,1]^2$ with covariance structure given in \eqref{covH}.
\end{thm}

\begin{Remark}  Theorem \ref{thm:sttest} suggests that under $H_0$, $g_u(\tau,\sigma) \equiv 0$, the quantity $K/B \times Q^{0}(\omega_1,\omega_2) = \frac{1}{B}\sum_{b=1}^B \|\mathcal{H}_b\|^2 + o(1)$. Therefore, approximate $p$-values for the test of stationarity can be constructed by comparing $K Q^{0}(\omega_1,\omega_2)/B$ with simulated quantiles of $\frac{1}{B}\sum_{b=1}^B \|\mathcal{H}_b\|^2,$ for large $B$. Note that  $\frac{1}{B}\sum_{b=1}^B \|\mathcal{H}_b\|^2$ is $O_p(1)$  as $B \to \infty$ by \autoref{lem:gb}. 
\end{Remark}
 The following Lemma guarantees consistency of the proposed test.
 
 \begin{lemma}\label{alt_st}
 Assume that $g_{u,\omega}(\tau,\sigma) = g_{u}(\tau,\sigma) \neq 0$ on a set of positive Lebesgue measure. Under \autoref{cumglsp} and \autoref{assymp}, 
 $K Q^{0}(\omega_1,\omega_2)/B \to \infty$ in probability. 
 \end{lemma}

\section{Finite Sample Properties}\label{sec:sec4}
\def\theequation{4.\arabic{equation}}
\setcounter{equation}{0}

\subsection{Simulation Studies} 
\label{sec:sims}
In order to evaluate the performance of the search algorithm in finite samples, we consider three simulation settings representing appropriate extensions of \citep{EBAsubmitted} for functional time series.  A B-spline basis with 15 basis functions is used to generate random realizations of the functional time series. 
All settings can be represented as $f_{u,\omega}(\tau,\sigma) = \phi_{u,\omega}f(\tau,\sigma)$ for $u\in [0,1]$ and $\omega \in (0,0.5)$ and   

\begin{equation}
\phi_{u,\omega} = 1 \; \mathrm{for} \; \omega \in
(0,0.5),
\end{equation}
\begin{equation}
\phi_{u,\omega} = 
\begin{cases}
10-9u \; \mathrm{for}  \; \omega \in (0,0.15 )\\
5 \; \mathrm{for}  \; \omega \in [0.15,0.35)\\
1+9u \; \mathrm{for} \; \omega \in [0.35, 0.5),\\
\end{cases}
\end{equation}
and 
\begin{equation}
\phi_{u,\omega} = 
\begin{cases}
2+\sin(8 \pi u - \pi/2)\; \mathrm{for}  \; \omega \in (0,0.15]\\
2+\cos(8 \pi u) \; \mathrm{for}  \; \omega \in (0.15,0.35]\\
2+\cos(16 \pi u) \; \mathrm{for} \; \omega \in (0.35, 0.5).\\
\end{cases}
\end{equation}
% \begin{equation}
% f_{u,\omega}(\tau,\sigma) = f(\tau,\sigma) \; \mathrm{for} \; \omega \in
% (0,0.5),
% \end{equation}
% \begin{equation}
% f_{u,\omega}(\tau,\sigma) = 
% \begin{cases}
% (10-9u)f(\tau,\sigma) \; \mathrm{for}  \; \omega \in (0,0.15)\\
% f(\tau,\sigma) \; \mathrm{for}  \; \omega \in [0.15,0.35)\\
% (1+9u)f(\tau,\sigma) \; \mathrm{for} \; \omega \in [0.35, 0.5),\\
% \end{cases}
% \end{equation}
% and 
% \begin{equation}
% f_{u,\omega}(\tau,\sigma) = 
% \begin{cases}
% (5.5+4.5\sin(8 \pi u - \pi/2))f(\tau,\sigma) \; \mathrm{for}  \; \omega \in (0,0.15]\\
% (5.5+4.5\cos(8 \pi u))f(\tau,\sigma) \; \mathrm{for}  \; \omega \in (0.15,0.35]\\
% (5.5+4.5\cos(16 \pi u))f(\tau,\sigma) \; \mathrm{for} \; \omega \in (0.35, 0.5),\\
% \end{cases}
% \end{equation}
See Figure \ref{fig:sims} for an illustration of the estimated time-varying auto spectrum for a single point in the functional domain for each setting.  In the first setting, we consider functional white noise in order to ensure the method maintains appropriate control for false positives.  In the second and third settings, both linear and non-linear nonstationary dynamics are considered within frequency bands.  These settings are specially designed to assess performance in detecting time-varying dynamics of different forms, as well as subtle changes in dynamics over frequencies.  

\begin{figure}[ht!]
    \centering
    \includegraphics[width=1\textwidth]{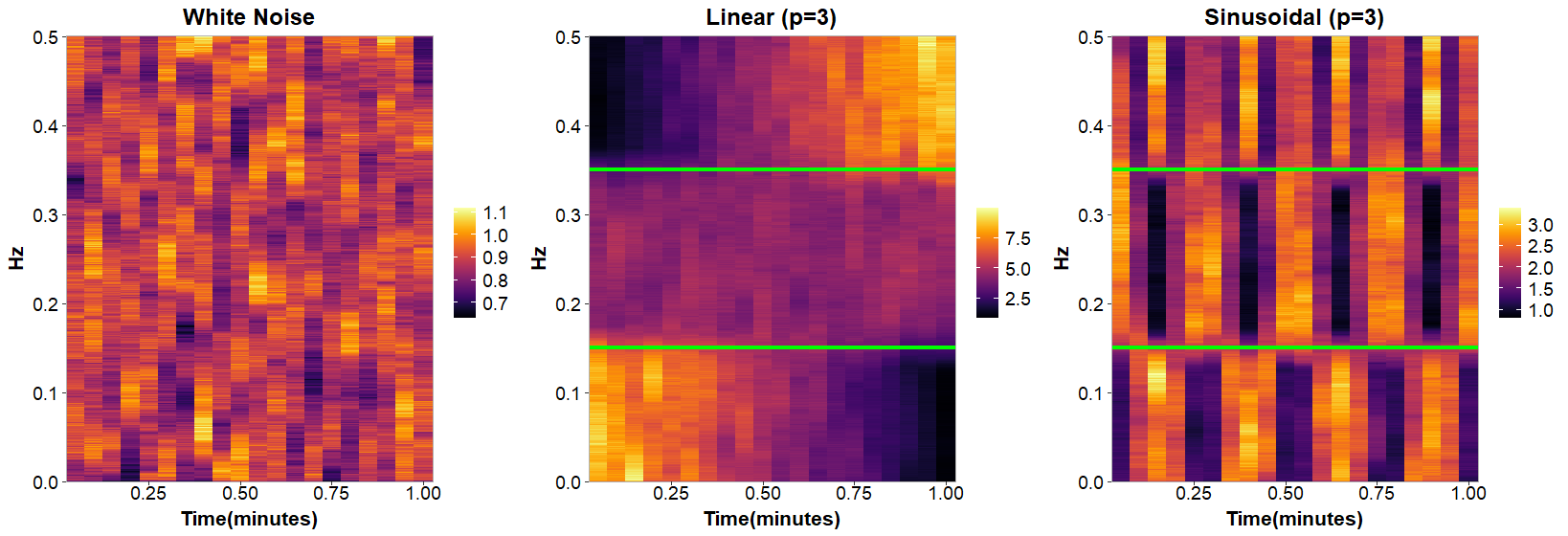}
    \caption{Local multitaper estimates of the time-varying autospectra for a single component from each of the simulation settings for time series with $N=$1000 observations per time segment and $B=$20 time segments.  Solid green lines represent the true partition of frequencies.}
    \label{fig:sims}
\end{figure}

Table \ref{tab:simresults} reports the means and standard deviations over 100 replications for the estimated number of frequency bands, $\hat{p}$, and Rand indices, $R(\hat{\boldsymbol{\omega}},\boldsymbol{\omega})$. The Rand index \citep{Rand1971} summarizes the similarity between the estimated frequency partition, $\hat{\boldsymbol{\omega}}$, and the true partition, $\boldsymbol{\omega}$.  Let $a$ be the number of pairs of Fourier frequencies in the same frequency band in $\hat{\boldsymbol{\omega}}$ and the same frequency band in $\boldsymbol{\omega}$ and $b$ be the number of pairs of Fourier frequencies in different frequency bands in $\hat{\boldsymbol{\omega}}$ and different frequency bands in $\boldsymbol{\omega}$.  Then the Rand index, $R(\hat{\boldsymbol{\omega}},\boldsymbol{\omega}) = (a+b)/\binom{N_B}{2}$, where $N_B=\lfloor T_B/2 \rfloor - 1$ is the number of Fourier frequencies.  The Rand index can take on values from 0 to 1 with values close to 1 indicating good estimation of the true frequency partition.  For these results, we consider performance under different combinations of settings for the number of observations in each time block, $T_B$, the number of approximately stationary time blocks, $B$, and the number of observations in the functional domain, $R$.  Furthermore, we fix the family-wise error rate (FWER) control level, $\alpha = 0.05$, number of frequencies tested in each pass, $n_{\max} = 30$, number of draws to approximate the null distribution of the test statistics $d_0=100,000$, local multitaper estimator bandwidth, $bw = 0.05$, such that the number of tapers, $K=\left \lfloor bw(N+1) \right \rfloor -1$, and we use a block diagonal approximation for the covariance function of the Gaussian process used to approximate the limiting null distribution of the scan statistics and select four approximately equally-spaced points in the functional domain for testing.  The proposed method provides good estimation accuracy for both the number and location of frequency band partition points, and performance generally improves as $T_B$ and $B$ increase.

\begin{table}[ht!]
\centering
\resizebox{4.5in}{!}{
\begin{tabular}{  c c| cc| cc  }

\hline 
 &\multicolumn{1}{c|}{} & \multicolumn{2}{ c| }{R=5} & \multicolumn{2}{ c}{R=10}\\

$T_B$ & \multicolumn{1}{c|}{B} & $\hat{p}$ & \multicolumn{1}{c|}{$R(\hat{\boldsymbol{\omega}},\boldsymbol{\omega})$} & $\hat{p}$ & $R(\hat{\boldsymbol{\omega}},\boldsymbol{\omega})$\\

 \hline
 \hline
\multicolumn{6}{ c }{White noise ($p=1$)} \\
\hline
 \hline
\multirow{3}{*}{200} 
 &  5   & 1.000(0.000) & 1.000(0.000) & 1.020(0.141) & 0.996(0.030)\\

 & 10   & 1.000(0.000) & 1.000(0.000) & 1.000(0.000) & 1.000(0.000)\\

 & 20  & 1.000(0.000) & 1.000(0.000) & 1.000(0.000) & 1.000(0.000)\\

 \hline
 
 \multirow{3}{*}{500} 
 &  5  & 1.970(0.810) & 0.655(0.248) & 1.850(0.796) & 0.701(0.250)\\

 & 10   & 1.350(0.539) & 0.862(0.209) & 1.260(0.525) & 0.900(0.195)\\

 & 20 & 1.130(0.367) & 0.944(0.155) & 1.130(0.367) & 0.948(0.146)\\

 \hline
 
\multirow{3}{*}{1000} 
&   5  & 1.800(0.841) & 0.735(0.255) & 1.780(0.786) & 0.736(0.251)\\

 & 10   & 1.060(0.239) & 0.974(0.105) & 1.150(0.411) & 0.941(0.157)\\

 & 20  & 1.020(0.141) & 0.992(0.059) & 1.010(0.100) & 0.996(0.042)\\

 \hline
 \hline
 
 \multicolumn{6}{ c }{Linear ($p=3$)} \\
\hline
\hline

\multirow{3}{*}{200} 
 &  5   & 2.120(0.498) & 0.738(0.129) & 2.210(0.498) & 0.762(0.115)\\

 & 10   & 2.200(0.471) & 0.761(0.102) & 2.240(0.495) & 0.767(0.105)\\

 & 20  & 2.460(0.501) & 0.816(0.083) & 2.530(0.502) & 0.826(0.082)\\

 \hline
 
 \multirow{3}{*}{500} 
 &  5  & 3.260(0.463) & 0.936(0.050) & 3.220(0.504) & 0.923(0.066)\\

 & 10   & 3.050(0.261) & 0.955(0.041) & 3.050(0.297) & 0.953(0.049)\\

 & 20 & 3.030(0.171) & 0.963(0.036) & 3.010(0.100) & 0.967(0.033)\\

 \hline
 
\multirow{3}{*}{1000} 
&   5  & 3.450(0.609) & 0.915(0.033) & 3.370(0.580) & 0.921(0.035)\\

 & 10   & 3.080(0.273) & 0.933(0.023) & 3.030(0.171) & 0.934(0.021)\\

 & 20  & 3.020(0.141) & 0.931(0.012) & 3.020(0.141) & 0.931(0.014)\\

 \hline
 \hline
 
 \multicolumn{6}{ c }{Sinusoidal ($p=3$)} \\
\hline
\hline

\multirow{3}{*}{200} 
 &  5   & 1.060(0.239) & 0.354(0.088) & 1.080(0.273) & 0.359(0.097)\\

 & 10   & 2.000(0.201) & 0.732(0.062) & 1.980(0.200) & 0.727(0.072)\\

 & 20  & 2.030(0.171) & 0.753(0.026) & 2.060(0.239) & 0.754(0.028)\\

 \hline
 
 \multirow{3}{*}{500} 
 &  5  & 2.600(0.865) & 0.726(0.143) & 2.590(0.805) & 0.718(0.153)\\

 & 10   & 2.950(0.458) & 0.897(0.085) & 2.910(0.534) & 0.893(0.088)\\

 & 20 & 3.040(0.243) & 0.952(0.039) & 2.990(0.266) & 0.945(0.054)\\

 \hline
 
\multirow{3}{*}{1000} 
&   5  & 2.300(0.916) & 0.651(0.182) & 2.140(0.954) & 0.620(0.192)\\

 & 10  & 2.660(0.623) & 0.829(0.106) & 2.720(0.533) & 0.854(0.102)\\

 & 20 & 3.040(0.197) & 0.932(0.016) & 3.010(0.100) & 0.934(0.012)\\

 \hline

\end{tabular}
}
\caption{Mean(standard deviation) for the estimated number of frequency bands, $\hat{p}$, and Rand index values, $R(\boldsymbol{\hat{\omega}},\boldsymbol{\omega})$, for $R=100$ replications.}
\label{tab:simresults}
\end{table}

 When the number of time blocks, $B$, is smaller, the proposed method slightly overestimates the number of partition points.  This is not unexpected, since the proposed method uses the asymptotic behavior of the scan statistics for testing purposes, which may not hold for smaller values of $B$.  For the first setting, Table \ref{tab:simresults} indicates good performance in correctly estimating only one frequency band.  Since the search algorithm is designed to control FWER, the false positive rate remains under control, even as the number of distinct frequencies tested increases with larger values of $T_B$.  In some cases, FWER may be controlled more tightly than $\alpha=0.05$ due to dependence among significance tests across frequencies \citep{efron2007,storey2007,causeuretal2009}. For the second setting and third settings, performance generally improves as $T_B$ and $B$ increase.  However, the accuracy is also impacted by the magnitude of the differences between the underlying demeaned time-varying power spectra across frequency bands, as well as the number of time blocks used to approximate the time-varying behavior. The time-varying dynamics of the power spectrum for adjacent frequency bands for the second setting are dissimilar for nearly all time points.  Accordingly, the algorithm is able to correctly estimate the frequency band partition with smaller $T_B$ and $B$ for this setting.  On the other hand, the two frequency bands covering higher frequencies in the third setting have similar time-varying dynamics in the power spectrum at particular time points due to their periodicities (see Figure \ref{fig:sims}). Also, without a sufficient number of time blocks, it is difficult to distinguish the different periodic time-varying behavior across frequency bands for this setting.  Taken together, this explains the need for relatively larger values of $T_B$ and $B$ for accurate estimation of the frequency band partition for the third setting compared to the second setting. 
%  For practical purposes, this also underscores the importance of making appropriate selections for the number of time blocks $B$ and number of tapers $K$ used to approximate the time-varying power spectrum.    

\subsection{Frequency Band Analysis for EEG Data}

To illustrate the usefulness of the proposed methodology for functional time series analysis, we turn to frequency band analysis of EEG signals. Analyzing EEG signals as nonstationary functional time series is warranted by the high-dimensionality, nonstationarity, and strong dependence typically observed for EEG signals.  Frequency bands are also commonly used in the scientific literature to generate summary measures of EEG power spectra, so a principled approach to frequency band estimation would be a welcomed development.  In the scientific literature, there is significant variability in frequency ranges used to define traditional EEG frequency bands \citep{banddefns}.  For the purposes of comparison with the proposed method, we use the following definitions (Hz) \citep{nacy2016controlling}: delta $(0,4)$, theta $[4,7)$, alpha $[7,12)$, beta $[12,30)$, and gamma $[30,100)$.  We analyze a 4-minute segment of a 72-channel 256 Hz EEG signal from a single participant from the study described in \citep{Trujilloetal2017}.  
% The signal was recorded at 256 Hz and frequencies above 52 Hz were removed using a low-pass filter. 
Participants in the study sat in a wakeful resting state and alternated between eyes open (EO) and eyes closed (EC) conditions at 1 minute intervals.  Given this particular design, we expect to see periodic nonstationary behavior similar to that of the third simulation setting introduced in Section \ref{sec:sims}. 

For computational efficiency, we downsample the time series at 32 Hz to produce a series of length $T=7680$.  For the local multitaper estimator of the time-varying power spectrum, we use 10-second segments resulting in $T_B=32 \times 10 = 320$ observations per segment and $B=7680/320 = 24$ segments.  Other parameters (e.g. FWER control level, number of frequencies tested in each pass, local multitaper bandwidth, etc.) follow the same settings used for the simulations in Section \ref{sec:sims}.  Since it is possible that different brain regions may be characterized by different frequency band structures, we consider two groups of channels, one representing the parietal and occipital lobes (PO7, PO3, POz, PO4, and PO8), which is associated with attention, and one representing the frontal and central lobes (F1, Fz, F2, FC1, FCz, FC2), which is associated with sensorimotor function \citep{eoecpaper}.  By applying the proposed method to each group, we can better understand if the frequency band structures that characterize the time-varying dynamics of the power spectrum are different within sub-regions of the functional domain.  

Figure \ref{fig:specPO} presents the log autospectra for the five channels associated with attentional system areas along with traditional EEG frequency bands and estimated EEG frequency bands using the proposed method.  
Applying the proposed methodology to this data revealed three frequency bands with different time-varying dynamics (Hz): (0,3.8), [3.8,12.6), [12.6,16).  Comparing with the traditional frequency band partition, the proposed low frequency band $(0,3.8)$ coincides with the traditional delta band $(0,4)$, but the next proposed frequency band $[3.8,12.6)$ covers both the traditional theta $[4,7)$ and alpha $[7,12)$ frequency bands.  This suggests that while the delta band exhibits different time-varying characteristics from other bands, the theta and alpha bands exhibit similar time-varying characteristics for these channels.   These results are not surprising, as it is well-known that alpha band power increases during EC conditions and attenuates with visual stimulation during EO conditions.  As is the case with this participant, similar behavior has also been observed for theta band power, with larger reductions observed in the posterior regions of the brain, including the attentional system area channels currently under study \citep{BARRY20072765,BARRY2017293}. Hence, it is reasonable that the time-varying dynamics of the alpha and theta bands may be similar for this particular set of EEG channels.  

\begin{figure}[ht!]
    \centering
    \includegraphics[width=1\textwidth]{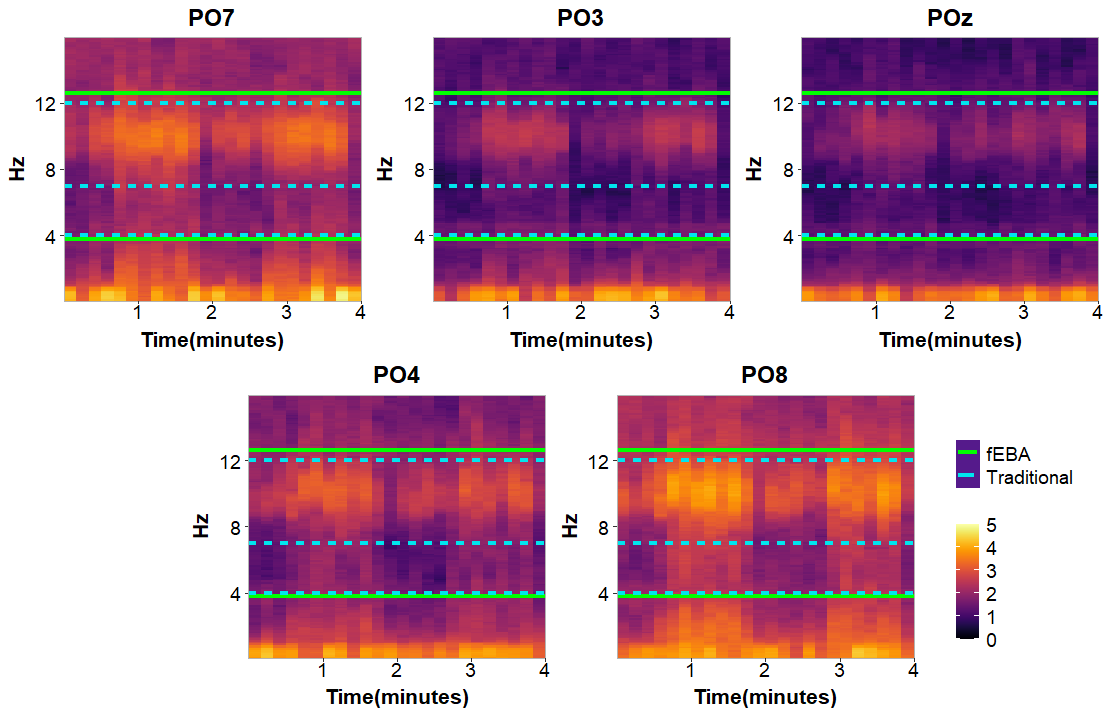}
    \caption{Log time-varying autospectra for 5 EEG channels and frequency bands determined by the proposed methodology (solid green lines) and traditional frequency bands (dashed blue lines) are displayed.  These EEG channels measure activity in the parietal (P) and occipital (O) lobes, which are associated with attention and visual processing.   }
    \label{fig:specPO}
\end{figure}

To better understand differences in the time-varying behavior across the estimated bands, Figure \ref{fig:ghatPO} displays a smoothed estimate of the frequency-band specific demeaned time-varying autospectra, $\hat{g}_u^{(p)}(\tau_i,\tau_i), p=1,2,3$ for each channel $i=1,\ldots,5$.  Smoothing was performed using cubic splines with 4 knots to better visualize the slowly-evolving time-varying dynamics under the assumption of local stationarity.  It can be seen that the time-varying dynamics for the estimated frequency band corresponding to the traditional delta band, (0, 3.8), coincides with the time-varying dynamics of the estimated frequency band covering the theta and alpha bands, [3.8, 12.6), for some channels (PO7, PO4, PO8), but not others (PO3, POz).  Since the search algorithm relies on an integrated scan statistic that integrates over the functional domain, it is sensitive to differences in the time-varying dynamics for proper subsets of the functional domain.  Accordingly, the proposed method correctly distinguishes between lower frequencies (0,3.8) and higher frequencies in the estimated frequency band structure.

\begin{figure}[ht!]
    \centering
    \includegraphics[width=1\textwidth]{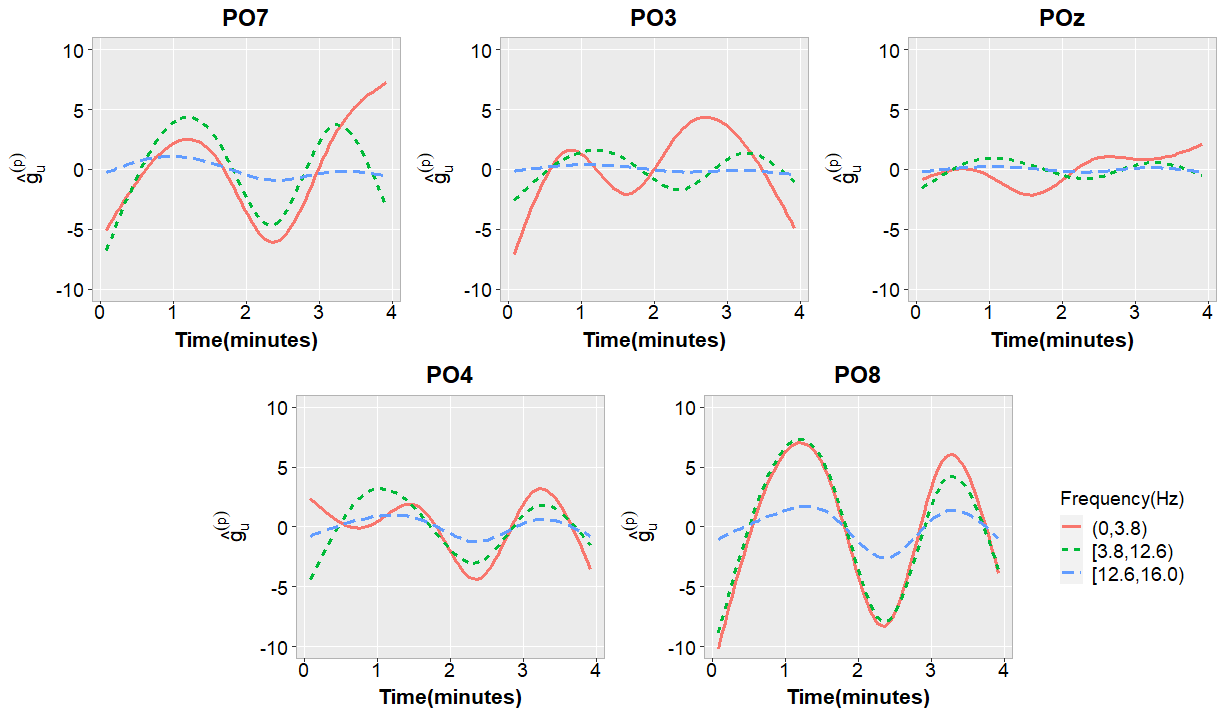}
    \caption{Smoothed estimator of the frequency band specific demeaned time-varying power spectra $g^{(p)}_u$ $p=1,2,3$ for 5 EEG channels using cubic splines with 4 knots.  Frequency bands are estimated used the proposed method and these channels measure activity in the parietal (P) and occipital (O) lobes.}
    \label{fig:ghatPO}
\end{figure}

Turning to the group of six EEG channels covering the frontal and central lobes, Figure \ref{fig:specFC} presents the log autospectra along with traditional EEG frequency bands and estimated EEG frequency bands using the proposed method.  
Applying the proposed methodology to this data revealed five frequency bands with different time-varying dynamics (Hz): (0,2.2), [2.2,4.9), [4.9, 8.1), [8.1,11.8), [11.8,16.0).  These estimated bands align reasonably well with the traditional EEG frequency bands.  However, the estimated bands suggest that the traditional delta band, (0,4), should be characterized by two sub-bands, (0,2.2) and [2.2, 4.9), which exhibit significantly different time-varying behavior of the power spectrum.  Such findings have been noted in the scientific literature, in which the so-called ``slow delta'' (0.7-2 Hz) and ``fast delta'' (2-4 Hz) bands exhibit different behavior during the wake-sleep transition \citep{BENOIT20002103}.  \citep{BARRY20072765} also observed that the magnitude of the difference in theta power between the EO and EC conditions is less for frontal and central brain regions compared to posterior regions, while the magnitude of the difference for the alpha band is similar across regions.  This is supported by the current analysis and can help explain why the theta and alpha bands are estimated to have similar time-varying behavior for the group of posterior region EEG channels (Figure \ref{fig:specPO}), and different time-varying dynamics for the group of central and frontal region EEG channels (Figure \ref{fig:specFC}).  

\begin{figure}[ht!]
    \centering
    \includegraphics[width=1\textwidth]{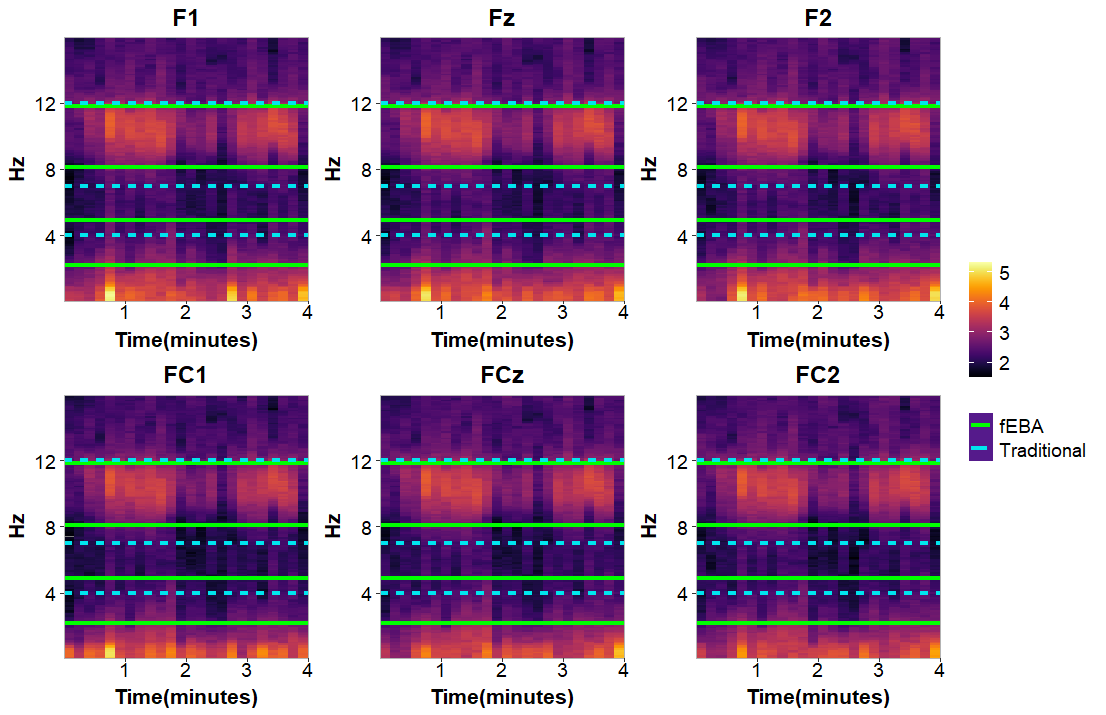}
    \caption{Log time-varying autospectra for 6 EEG channels and frequency bands determined by the proposed methodology (solid green lines) and traditional frequency bands (dashed blue lines) are displayed.  These EEG channels measure activity in the frontal (F) and central (C) lobes, which are associated with sensorimotor function.}
    \label{fig:specFC}
\end{figure}

The smoothed estimates of the frequency-band specific demeaned time-varying autospectra for these channels (see Figure \ref{fig:ghatFC}) illustrate the different time-varying behavior of the estimated frequency bands captured by the search algorithm.  The two sub-bands covering the traditional delta band indicate very different time-varying behavior.  Also, the estimated band that roughly corresponds to the traditional alpha band, [8.1, 11.8), has a more regular and pronounced time-varying behavior, corresponding to the alternating EO and EC conditions, compared to the other estimated frequency bands. In summary, the proposed search algorithm estimates frequency bands that can better characterize the power spectrum for the particular functional EEG time series under study, compared to traditional EEG frequency band analysis, and can be used to construct customized frequency band summary measures for characterizing different brain regions.   

\begin{figure}[ht!]
    \centering
    \includegraphics[width=1\textwidth]{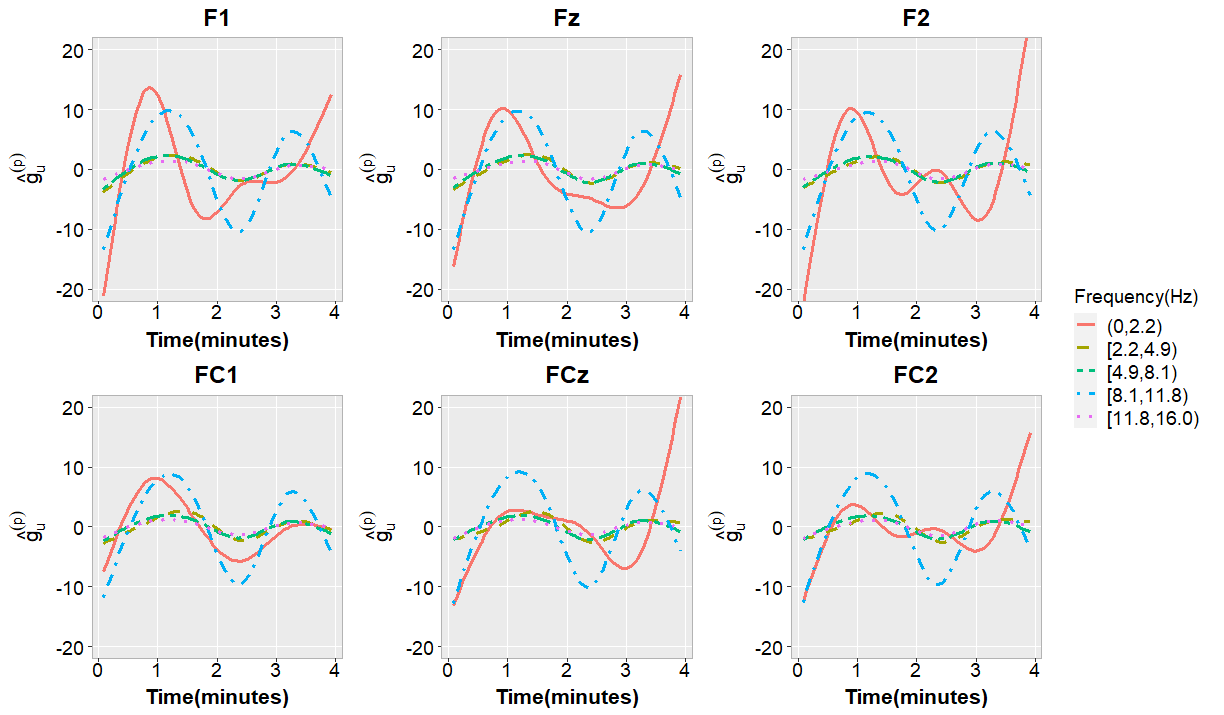}
    \caption{Smoothed estimator of the frequency band specific demeaned time-varying power spectra $g^{(p)}_u$ $p=1,2,3$ for 5 EEG channels using cubic splines with 4 knots.  Frequency bands are estimated used the proposed method and these channels measure activity in the frontal (F) and central (C) lobes.}
    \label{fig:ghatFC}
\end{figure}

\section{Discussion}
\label{sec:discussion}
The frequency band analysis framework for nonstationary functional time series introduced in this article offers a quantitative approach to identifying frequency bands that best preserve the nonstationary dynamics of the underlying functional time series.  This framework allows for estimation of both the number of frequency bands and their corresponding endpoints through the use of a sensible integrated scan statistic within an iterative search algorithm.  Another test statistic is also offered to determine which bands, if any, are stationary with respect to time.  
Motivated by the application to EEG frequency band analysis, it would be interesting to consider extensions of this framework that enable localization of the frequency band estimation framework in the time and functional domains.  
Such extensions would allow for the frequency band estimation framework to automatically adapt to local spectral characteristics without needing to pre-specify particular time segments or subsets of the functional domain for analysis.  However, these extensions present significant computational challenges associated with searching over multiple spaces simultaneously. 
% As seen through the application, the frequency band structure most appropriate for summarizing the characteristics of the power spectrum may vary across the functional domain.  
% Rather than the integrated scan statistics proposed in this work, functional scan statistics could be developed and the search algorithm could be extended to allow for exploration of the frequency and functional domains.  Extending this framework for time-dependent frequency band estimation would also be an interesting future direction of research to determine appropriate frequency band structures for particular points in time that may correspond to changing underlying conditions, such as the alternating EO and EC conditions in the application.  

% It would be interesting to consider other forms of the time-varying power spectrum estimator, such as unequally-sized or overlapping segment piecewise stationary approximations, or smooth estimators, such as wavelet-based estimators, as potential enhancements to this framework.  However, as mentioned previously, the distributional properties of the FRESH statistics in these settings may be nontrivial.  Another interesting research direction would be to extend the frequency domain partitioning framework to encompass both time and frequency domains.  This could offer a unified solution that can identify time-localized frequency bands for characterizing time-varying power spectra.  
% \sab{added last two sentences to appease reviewer.  Thoughts?} 

We have focused on estimation of frequency bands for a single nonstationary functional time series, but this framework can also be extended for the analysis of multiple functional time series.  For example, extending this framework for estimating frequency bands for classification and clustering of functional time series would provide researchers with optimal frequency band features for supervised and unsupervised learning tasks.  This extension could be very useful in the study of EEG and fMRI signals to construct frequency band features that are associated with clinical and behavioral outcomes or that can be used to identify groups of time series with similar spectral characteristics.

%%%%%%%%%%%%%%%%%%%%%%%%%%%%%%%%%%%%%%%%%%%%%%
%% Single Appendix:                         %%
%%%%%%%%%%%%%%%%%%%%%%%%%%%%%%%%%%%%%%%%%%%%%%
%\begin{appendix}
%\section*{???}%% if no title is needed, leave empty \section*{}.
%\end{appendix}
%%%%%%%%%%%%%%%%%%%%%%%%%%%%%%%%%%%%%%%%%%%%%%
%% Multiple Appendixes:                     %%
%%%%%%%%%%%%%%%%%%%%%%%%%%%%%%%%%%%%%%%%%%%%%%
\begin{appendix}
\def\theequation{A.\arabic{equation}}
\setcounter{equation}{0}
\section*{Technical Details}

\subsection{Properties of Multitaper Periodogram of Functional Time Series}

We first establish some asymptotic properties of the multitaper periodogram estimator defined in \eqref{eq:multpdg}. Following the notations in \citep{dahlhaus1985asymptotic} we define the quantities
$$H_{k}(\omega) =\sum_{t=1}^{T_B} e^{-2\pi\omega t}v_b^k(t),$$
$$H_{k,l}(\omega) =\sum_{t=1}^{T_B} e^{-2\pi\omega t}v_b^k(t)v_b^l(t),$$
%Note that $H_k(\omega) = O(1)$ if $\omega = 0 \mod 2\pi$ and $O(1/T_B)$ otherwise. 
Note that, 
$$\int_{-1}^{1}\vert H_k(\omega)\vert^2d\omega  = \sum_{t=1}^{T_B} (v_b^k(t))^2 = 1.$$
Moreover by \citep{dahlhaus1985asymptotic} we can write an upper bound $L(\omega)$ for both the functions $\sqrt{T_B}H_k(\omega)$ and  ${T_B}H_{k,k}(\omega)$ where
$$L(\omega) = \left\{\begin{array}{ll}
{T_B},     & \text{if }\vert \omega \vert \leq 1/T_B  \\
\frac{1}{\vert \omega \vert}, &     \text{otherwise}.
\end{array}\right. $$

\begin{lemma}
\label{cumfft} Let $\tilde{X}_T^{(k),b,\omega}$ be the $k$-th functional discrete Fourier transformation of the observed time series at $b$-th block and frequency $\omega$, as defined in \eqref{eq:fDFT}. Under \autoref{cumglsp} and \autoref{assymp}, the $l$-th order cumulant kernel of the fDFT is given by
\begin{align*}
    &cum\left(\tilde{X}_T^{(k_1),b_1,\omega_1}(\tau_1),\dots,\tilde{X}_T^{(k_l),b_l,\omega_l}(\tau_l)\right)\\
    &=\left\{\begin{array}{ll}
        H_{k_1,\dots,k_l}\left(\sum_{j=1}^l \omega_j\right)f_{u_b,\omega_1,\dots,\omega_{l-1}}(\tau_1,\dots,\tau_l) + o(1), & \text{if } b_1=b_2=\dots=b_l=b\\ 0, & \text{otherwise. } 
    \end{array}\right.
\end{align*}
where $H_{k_1,\dots,k_l}(\omega) = \displaystyle\sum_{t=1}^{T_B}v_b^{k_1}(t)v_b^{k_2}(t)\dots v_b^{k_l}(t)e^{-i\omega t}.$
\end{lemma}

\begin{proof}
Let $B_b$ be the $b$-th time block and write
\begin{align*}
    & cum\left(\tilde{X}_T^{(k_1),b_1,\omega_1}(\tau_1),\dots,\tilde{X}_T^{(k_l),b_l,\omega_l}(\tau_l)\right)\\
    = & cum\left(\sum_{t \in B_{b_1}} v_{b_1}^{k_1}(t)X_{t,T}(\tau_1)e^{-\im 2\pi \omega_1 t},\dots,\sum_{t \in B_{b_l}} v_{b_l}^{k_l}(t)X_{t,T}(\tau_l)e^{-\im 2\pi \omega_l t}\right)\\
    = &  \sum_{t_1 \in B_{b_1}}\dots \sum_{t_l \in B_{b_l} } v_{b_1}^{k_1}(t_1)\dots v_{b_l}^{k_l}(t_l)  e^{-\im 2\pi \sum_{k=1}^l \omega_k t_k} cum \left(X_{t_1,T}(\tau_1),\dots,X_{t_l,T}(\tau_l)\right)\\
    = &  \sum_{t_1 \in B_{b_1}}\dots \sum_{t_l \in B_{b_l} } v_{b_1}^{k_1}(t_1)\dots v_{b_l}^{k_l}(t_l)  e^{-\im 2\pi \sum_{k=1}^l \omega_k t_k} cum \left(X_{t_1}^{(u_{b_1})}(\tau_1),\dots,X_{t_l}^{(u_{b_l})}(\tau_l)\right) +o(1).
\end{align*}
If we have $b_i \neq b_j$ for any pair $i,j \in \{1,2,\dots,l\}$, the last expression converges to 0 by \autoref{cumglsp}.

For the case $b_1=b_2=\dots=b_l = b$, we write $t_j = t_1 + v_j$ for $j=2,\dots,l$. The expression is then simplified to
\begin{align*}
    &\sum_{v_2}\dots\sum_{v_l}\exp\left(-\im\sum_{j=2}^{l}v_j\omega_j\right) cum \left(X_{t_1}^{(u_{b})}(\tau_1),\dots,X_{t_1+v_l}^{(u_b)}(\tau_l)\right)\sum_{t_1}v_{b}(t_1)v_b(t_1+v_2) \times\\
    &~~~~~~~~~~\dots \times v_b(t_1+v_l)\exp\left(-\im \sum_{j=1}^l \omega_j t_1\right)
\end{align*}
The result then follows by Lemma P.4.1 an Lemma P.4.2 from \citep{brillinger2001}.
\end{proof}

\begin{lemma}\label{lem:multpdg}
Let $\widehat{f}^{(mt)}_{b,\omega}$ be the multitaper periodogram estimator, as defined in \eqref{eq:multpdg}. Under \autoref{assymp}, we have 
\begin{align}
    &\E \hat{f}^{(mt)}_{b,\omega}(\tau,\sigma) = f_{u_b,\omega}(\tau,\sigma) + O\left(\log(T_B)/T_B\right) + O(1/T).\\
    &\text{Cov}\left(\hat{f}^{(mt)}_{b,\omega_1}(\tau_1,\sigma_1),\hat{f}^{(mt)}_{b,\omega_2}(\tau_2,\sigma_2)\right) =  \displaystyle{\left\{\begin{array}{ll}
       \frac{f_{u_b,\omega}(\tau_1,\tau_2)f_{u_b,\omega}(\sigma_1,\sigma_2)}{K}  + O(1/T_B) + O(1/T), & \text{ if } \omega_1 = \omega_2 = \omega \nonumber \\
       O(1/T_B) + O(1/T), & \text{ if } \omega_1 \neq \omega_2.
    \end{array}
    \right.}
\end{align}
\end{lemma}

\begin{proof}
We start by noting that,
\begin{equation*}
    \E \hat{f}^{(k)}_{b,\omega}(\tau,\sigma) = \int_{-1}^{1}\vert H_k(\alpha)\vert^2f_{u_b,\omega-\alpha}d\alpha + O(1/T),
\end{equation*}
where $u_b$ is the midpoint of the $b$-th segment. This follows from Theorem 5.2.3 from \citep{brillinger2001} applied to the approximating series $X^{(u_b)}_t$ within the $b$-th block.
Now write $f_{u_b,\omega-\alpha} = f_{u_b,\omega} + O(\vert \alpha \vert)$ and using the form of $L$, we have
\begin{equation*}
    \E \hat{f}^{(k)}_{b,\omega}(\tau,\sigma) = f_{u_b,\omega}(\tau,\sigma) + O\left(\log(T_B)/T_B\right) + O(1/T).
\end{equation*}
Taking an average of the last expression over different tapers, we get
\begin{equation*}
    \E \hat{f}^{(mt)}_{b,\omega}(\tau,\sigma) = f_{u_b,\omega}(\tau,\sigma) + O\left(\log(T_B)/T_B\right) + O(1/T).
\end{equation*}
To calculate the covariance, first note that
$$\text{Cov}\left(\hat{f}^{(k)}_{b_1,\omega_1}(\tau_1,\sigma_1),\hat{f}^{(k)}_{b_2,\omega_2}(\tau_2,\sigma_2)\right) = 0,~~~~~ \text{if } b_1 \neq b_2.$$
Similar calculations as in the proof of Theorem 5.2.8 in \citep{brillinger2001} yields  
\begin{align*} 
    &\text{Cov}\left(\hat{f}^{(k)}_{b,\omega_1}(\tau_1,\sigma_1),\hat{f}^{(k)}_{b,\omega_2}(\tau_2,\sigma_2)\right)\\ =& \left[\vert H_{k,k}(\omega_1 + \omega_2)\vert^2 f_{u_b,\omega_1}(\tau_1,\sigma_2)f_{u_b,\omega_2}(\tau_2,\sigma_1) + \vert H_{k,k}(\omega_1 - \omega_2) \vert^2 f_{u_b,\omega_1}(\tau_1,\tau_2)f_{u_b,\omega_2}(\sigma_1,\sigma_2)\right]\\
    &~~~~~~+ O(1/T_B) + O(1/T).
\end{align*}
Therefore if $\omega_1 \neq \omega_2$ we have
$$\text{Cov}\left(\hat{f}^{(k)}_{b,\omega_1}(\tau_1,\sigma_1),\hat{f}^{(k)}_{b,\omega_2}(\tau_2,\sigma_2)\right) = O(1/T_B) + O(1/T),$$
and
$$\text{Cov}\left(\hat{f}^{(k)}_{b,\omega}(\tau_1,\sigma_1),\hat{f}^{(k)}_{b,\omega}(\tau_2,\sigma_2)\right) = f_{u_b,\omega}(\tau_1,\tau_2)f_{u_b,\omega}(\sigma_1,\sigma_2) + O(1/T_B) + O(1/T).$$
For local periodograms calculated for different tapers,
\begin{align*}
    &\text{Cov}\left(\hat{f}^{(k)}_{b,\omega_1}(\tau_1,\sigma_1),\hat{f}^{(l)}_{b,\omega_2}(\tau_2,\sigma_2)\right)\\ =& \left[\vert H_{k,l}(\omega_1 + \omega_2)\vert^2 f_{u_b,\omega_1}(\tau_1,\sigma_2)f_{u_b,\omega_2}(\tau_2,\sigma_1) + \vert H_{k,l}(\omega_1 - \omega_2) \vert^2 f_{u_b,\omega_1}(\tau_1,\tau_2)f_{u_b,\omega_2}(\sigma_1,\sigma_2)\right]\\
    &~~~~~~+ O(1/T_B) + O(1/T).
\end{align*}
By the orthogonality of the tapers,
$H_{k,l}(0) = 0$ for $k \neq l.$ And by Cauchy-Schwartz, we have $H_{k,l}(\omega) \leq \sqrt{H_{k,k}(\omega)}\sqrt{H_{l,l}(\omega)}$, and hence $L$ is indeed an upper bound for $T_B H_{k,l}(\omega).$ Therefore, in general for $k \neq l$,
\begin{align*}
    &\text{Cov}\left(\hat{f}^{(k)}_{b,\omega_1}(\tau_1,\sigma_1),\hat{f}^{(l)}_{b,\omega_2}(\tau_2,\sigma_2)\right) = O(1/T_B) + O(1/T).
\end{align*}
Therefore we have,
\begin{align*}
    &\text{Cov}\left(\hat{f}^{(mt)}_{b,\omega_1}(\tau_1,\sigma_1),\hat{f}^{(mt)}_{b,\omega_2}(\tau_2,\sigma_2)\right) = \frac{1}{K^2}\sum_{k=1}^K \sum_{l=1}^K \text{Cov}\left(\hat{f}^{(k)}_{b,\omega_1}(\tau_1,\sigma_1),\hat{f}^{(l)}_{b,\omega_2}(\tau_2,\sigma_2)\right)\\
    = & 
       \frac{1}{K} \left[ f_{u_b,\omega_1}(\tau_1,\tau_2)f_{u_b,\omega_2}(\sigma_1,\sigma_2) + f_{u_b,\omega_1}(\tau_1,\sigma_2)f_{u_b,\omega_2}(\tau_2,\sigma_1) \right]+ O(1/T_B) + O(1/T).
 \end{align*}
\end{proof}

The next theorem establishes a Central Limit Theorem type result for the multitaper periodogram estimator and is essential for proving \autoref{thm:null} and \autoref{thm:alt}.

\begin{thm}\label{clt_multpdg}
Consider the processes $E_{b,j} \in L^2[0,1]^2$ defined as $$E_{b,j}(\tau,\sigma) = \sqrt{K}\left(\hat{f}^{(mt)}_{b,\omega_j}(\tau,\sigma) - f_{u_b,\omega_j}(\tau,\sigma) \right)$$ for $b=1,2,\dots,B$ and $j=1,2,\dots,J$. For fixed $B$, as $T \to \infty$, $K \to \infty$ and $T/BK \to \infty$, the finite dimensional distributions of $\{E_{b,j}\}_{j,b}$ converge to a multivariate normal distribution. More precisely, for all $(\tau_1,\sigma_1),\dots,(\tau_d,\sigma_d) \in [0,1]^2$ and for all $d \in \mathbb{N}$,
$$\{E_{b,j}(\tau_1,\sigma_1),\dots,E_{b,j}(\tau_d,\sigma_d)\}_{b,j} \stackrel{d}{\to} \{Z_{b,j}(\tau_1,\sigma_1),\dots, Z_{b,j}(\tau_d,\sigma_d)\}_{b,j}$$
where $\{Z_{b,j}(\tau_1,\sigma_1),\dots, Z_{b,j}(\tau_d,\sigma_d)\}_{b,j}$ is a multivariate normal random vector with zero mean and covariance structure 
\begin{align*}
    &\text{Cov}\left(Z_{b_1,j_1}(\tau_1,\sigma_1),Z_{b_2,j_2}(\tau_2,\sigma_2)\right)\\
    &= \left\{\begin{array}{ll}
       f_{u_b,\omega_{j_1}}(\tau_1,\tau_2)f_{u_b,\omega_{j_2}}(\sigma_1,\sigma_2) + f_{u_b,\omega_{j_1}}(\tau_1,\sigma_2)f_{u_b,\omega_{j_2}}(\tau_2,\sigma_1) ,  & \text{if }  ~b_1 = b_2 = b,  \\
        0, & \text{otherwise.} 
    \end{array}\right.
\end{align*}
\end{thm}

\begin{proof}
We will show the cumulants of the vector $\{E_{b,j}(\tau_1,\sigma_1),\dots,E_{b,j}(\tau_d,\sigma_d)\}_{b,j}$ converge to the cumulants of the vector $\{Z_{b,j}(\tau_1,\sigma_1),\dots, Z_{b,j}(\tau_d,\sigma_d)\}_{b,j}$. As the cumulants of order $l$ of the Gaussian distribution are zero for $l >2$, we will show that as $T \to \infty$ and $K \to \infty$,
\begin{align*}
  &cum\left(E_{b_1,j_1}(\tau_1,\sigma_1),\dots,\dots,E_{b_l,j_l}(\tau_l,\sigma_l)\right)\\ &= \left\{\begin{array}{ll}
    o(1) & \text{if } l \neq 2  \\
    \text{Cov}\left(Z_{b_1,j_1}(\tau_1,\sigma_1),Z_{b_2,j_2}(\tau_2,\sigma_2)\right) + o(1) & \text{if } l =2. 
\end{array} \right.  
\end{align*}
Note that the equality for $l=1$ and $2$ follow from the earlier expectation and variance calculation. Hence we have to show the result for $l \geq 3.$\\
\\
Note that
\begin{align*}
     &cum\left(E_{b_1,j_1}(\tau_1,\sigma_1),\dots,\dots,E_{b_l,j_l}(\tau_l,\sigma_l)\right)\\ &=K^{l/2}cum\left(\hat{f}^{(mt)}_{b_1,\omega_{j_1}}(\tau_1,\sigma_1),\dots,\dots,\hat{f}^{(mt)}_{b_l,\omega_{j_l}}(\tau_l,\sigma_l)\right)\\
     & = \frac{1}{K^{l/2}}\sum_{k_1}\dots\sum_{k_l}cum\left(\hat{f}^{(k_1)}_{b_1,\omega_{j_1}}(\tau_1,\sigma_1),\dots,\dots,\hat{f}^{(k_
     l)}_{b_l,\omega_{j_l}}(\tau_l,\sigma_l)\right)\\
     & = \frac{1}{K^{l/2}}\sum_{k_1}\dots\sum_{k_l}cum\left(Y_{11}Y_{12},\dots,Y_{l1}Y_{l2}\right)
\end{align*}
where $Y_{i1} = \widetilde{X}_T^{(k_i),b_i,\omega_i}(\tau_i)$ and $Y_{i2} = \widetilde{X}_T^{(k_i),b_i,\omega_i}(\sigma_i)$.
\\
Using Theorem 2.3.2 from \citep{brillinger2001}, the last quantity is equal to
$$\frac{1}{K^{l/2}}\sum_{k_1}\dots\sum_{k_l} \sum_{\nu} cum(Y_{ij}:ij \in \nu_1)\dots cum(Y_{ij}:ij \in \nu_p) =:  \sum_{\nu} C(\nu)$$
where the sum is over all indecomposable partitions $\nu = \nu_1 \cup \nu_2 \cup \dots \cup \nu_p$ of
$$\begin{array}{cc}
   (1,1)  & (1,2)  \\
   (2,1)  & (2,2) \\
   \vdots & \vdots\\
   (l,1) & (l,2).
\end{array}$$
As there are a finite number of partitions, it is enough to show $C(\nu) = o(1)$ for all indecomposible partitions $\nu$ for $l>2$. \\
\\
To this end, we note that the function $H_{k_1,k_2,\dots,k_m}(\omega) = O(T_B^{-m/2})$ if $\omega \neq 0.$ By the orthonormality and symmetry of the wave function,
$H_{k_1,k_2,\dots,k_m}(\omega) = O(T_B^{1-m/2})$ if $\omega = 0$ and all distinct $k_i$'s appear an even number of times in the index $k_1,\dots,k_m$, and $H_{k_1,k_2,\dots,k_m}(\omega) = 0$, if $\omega=0$ and any $k_i$ appear an odd number of times.\\
\\
 Let $\vert \nu_i\vert$ denote the number of elements of the set $\nu_i$. Note that $\nu$ is a partition of a set of $2l$ elements, and therefore $\sum_{i=1}^p \vert \nu_i \vert = 2l.$ By Lemma \ref{cumfft} and the property of the $H_{k_1,\dots,k_l}$ functions, we note that
 $C(\nu) = 0$ if any $\nu_i$ in $\nu$ has at least one $k_i$ an odd number of times and otherwise 
$$C(\nu) = \frac{1}{K^{l/2}}\sum_{k_1}\dots \sum_{k_l}\prod_{i=1}^p O\left(T_B^{-\vert \nu_i \vert/2 + 1\left(\sum_{j \in \nu_i}\omega_j = 0 \mod 2\pi\right)}\right) = \sum_{r=1}^{l} O(K^{r-l/2})O\left(T_B^{-l + s(\nu)}\right),$$
where $r$ is the distinct number of $k_i$'s in a collection $k_1,k_2,\dots,k_l$ and
$s(\nu)$ = the number of $\nu_i$ in $\nu$, such that $\sum_{j \in \nu_i}\omega_j = 0.$\\
\\
If $r > l/2$, then at least one of the $k_i$'s appear just once, and therefore one of the sets in the partition must have a single occurrence of that index. So it is enough to consider the case where $r \leq l/2$. Now consider the possibilities for the $O\left(T_B^{-l+s(\nu)}\right)$ term.\\ 
\\
\textbf{Case 1:} If $p < l$ then $s(\nu) \leq p < l$, and therefore $O\left(T_B^{-l+s(\nu)}\right) = o(1).$\\
\\
\textbf{Case 2:} If $p > l$, at least $2(p-l)$ sets of the partitions have just one element. (To see this, suppose $l_1$ is the number of sets with one element. Then we have $2l \geq l_1 + 2(p-l_1)$.) For all those one element sets $\sum_j \omega_j \neq 0 \mod 2\pi.$ Therefore $s(\nu) \leq p - 2(p-l) = 2l - p < l$ and consequently $O\left(T_B^{-l+s(\nu)}\right) = o(1).$
\\
\\
\textbf{Case 3:} If $p = l$ and at least one set in the partition has a single element, for that set $\sum_j \omega_j \neq 0 \mod 2\pi$, therefore $s(\nu) \leq l-1$ and $O\left(T_B^{-l+s(\nu)}\right) = o(1).$ \\
\\
\textbf{Case 4:} Consider the case where $p=l$ and all the partitions have 2 elements. Note that as $s(\nu)\leq p$, if $r < l/2$ the product $O(K^{r-l/2})O\left(T_B^{-l+s(\nu)}\right) = o(1).$ Therefore it is enough to consider the case where $l$ is even and $r = l/2.$ If $l >2$, this means there are at least two distinct $k_i$ in the collection and by indecomposibility, one of the sets $\nu_i$ in the partition must have two distinct $k_i$, making $C(\nu) = 0$.
\end{proof}

\subsection{Proof of Results in Section 3}

\subsubsection*{Proof of \autoref{lem:ghat}} Noting the definition of $\widehat{g}_{u,\omega}$ in \eqref{eq:demps}, the result follows from \autoref{lem:multpdg} and some simple algebra.

\begin{lemma}
\label{clt_g}
Consider the processes $H_{b,j} \in L^2[0,1]^2$ defined as $$H_{b,j}(\tau,\sigma) = \sqrt{K}\left(\widehat{g}_{b/B,\omega_j}(\tau,\sigma) - g_{u_b,\omega_j}(\tau,\sigma) \right)$$ for $b=1,2,\dots,B$ and $j=1,2,\dots,J$. For fixed $B$, as $T \to \infty$, $K \to \infty$ and $T/BK \to \infty$, the finite dimensional distributions of $\{H_{b,j}\}_{j,b}$ converges to a multivariate normal distribution. More precisely, for all $(\tau_1,\sigma_1),\dots,(\tau_d,\sigma_d) \in [0,1]^2$ and for all $d \in \mathbb{N}$
$$\{H_{b,j}(\tau_1,\sigma_1),\dots,H_{b,j}(\tau_d,\sigma_d)\}_{b,j} \stackrel{d}{\to} \{Z'_{b,j}(\tau_1,\sigma_1),\dots, Z'_{b,j}(\tau_d,\sigma_d)\}_{b,j}$$
where $\{Z'_{b,j}(\tau_1,\sigma_1),\dots, Z'_{b,j}(\tau_d,\sigma_d)\}_{b,j}$ is a multivariate normal random vector with zero mean and covariance structure 
\begin{align}\label{ghat_cov}
    &\text{Cov}\left(Z'_{b,j_1}(\tau_1,\sigma_1),Z'_{b,j_2}(\tau_2,\sigma_2) \right)=  \left(1 - \frac{2}{B}\right)F(u_b,\omega_1,\omega_2,\tau_1,\sigma_1,\tau_2,\sigma_2)\nonumber\\
    & ~~ \hspace{2.25 in}+ \frac{1}{B^2}\sum_{l=1}^B F(u_l,\omega_1,\omega_2,\tau_1,\sigma_1,\tau_2,\sigma_2).\nonumber\\
    &\text{Cov}\left(Z'_{b_1,j_1}(\tau_1,\sigma_1),Z'_{b_2,j_2}(\tau_2,\sigma_2) \right)=  -\frac{1}{B}F(u_{b_1},\omega_1,\omega_2,\tau_1,\sigma_1,\tau_2,\sigma_2)\nonumber\\
   &~~~~~~~~~~~~~~~~~~~~ - \frac{1}{B} F(u_{b_2},\omega_1,\omega_2,\tau_1,\sigma_1,\tau_2,\sigma_2) + \frac{1}{B^2}\sum_{l=1}^B F(u_{b_l},\omega_1,\omega_2,\tau_1,\sigma_1,\tau_2,\sigma_2),
\end{align}
where $b_1 \neq b_2$ and
\begin{equation}\label{F_cov}
    F(u,\omega_1,\omega_2,\tau_1,\sigma_1,\tau_2,\sigma_2) := f_{u,\omega_1}(\tau_1,\tau_2)f_{u,\omega_2}(\sigma_1,\sigma_2) + f_{u,\omega_2}(\tau_1,\sigma_2)f_{u,\omega_2}(\tau_2,\sigma_1).
\end{equation}

% \begin{align*}
%     &\text{Cov}\left(Z'_{b_1,j_1}(\tau_1,\sigma_1),Z'_{b_2,j_2}(\tau_2,\sigma_2)\right)\\
%     &= \left\{\begin{array}{ll}
%       f_{u_b,\omega_{j_1}}(\tau_1,\tau_2)f_{u_b,\omega_{j_2}}(\sigma_1,\sigma_2) + f_{u_b,\omega_{j_1}}(\tau_1,\sigma_2)f_{u_b,\omega_{j_2}}(\tau_2,\sigma_1) ,  & \text{if }  ~b_1 = b_2 = b,  \\
%         0, & \text{otherwise.} 
%     \end{array}\right.
% \end{align*}
\end{lemma}

\begin{proof}
In view of \autoref{lem:ghat} it is enough to show that the joint cumulants of $\{H_{b,j}(\tau_1,\sigma_1),\dots,H_{b,j}(\tau_d,\sigma_d)\}_{b,j}$ for $b=1,2,\dots,B$ and $j=1,2,\dots,J$ of order $> 2$ converges to $0$, as $T \to \infty,$ $K \to \infty$ and $T/K \to \infty$. Note that by definition of $\widehat{g}$, we have
$$H_{b,j}(\tau,\sigma) = E_{b,j}(\tau,\sigma) - \frac{1}{B} \sum_{b=1}^B E_{b,j}(\tau,\sigma).$$
Therefore using the linearity of cumulants (Theorem 2.3.1 (i) \& (iii) from \citep{brillinger2001}) we can write 
$cum\left(H_{b_1,j_1}(\tau_1,\sigma_1),\dots,\dots,H_{b_l,j_l}(\tau_l,\sigma_l)\right)$ as sum of $2^l$ terms, where each term is of the form
$$cum\left(\widetilde{E}_{b_1,j_1}(\tau_1,\sigma_1),\dots,\dots,\widetilde{E}_{b_l,j_l}(\tau_l,\sigma_l)\right)$$
where $\widetilde{E}_{b,j}(\tau,\sigma)$ is either ${E}_{b,j}(\tau,\sigma)$ or the average $\frac{1}{B}\sum_{b=1}^B {E}_{b,j}(\tau,\sigma)$.
Without loss of generality, consider a term where the first $k$  $\widetilde{E}_{b,j}(\tau,\sigma)$'s are the averages and the last  $(l-k)$ are ${E}_{b,j}$. That particular term can then be simplified to
$$\frac{1}{B^k}\sum_{b_1=1}^B\dots\sum_{b_k=1}^B cum\left(E_{b_1,j_1}(\tau_1,\sigma_1),\dots,\dots,E_{b_l,j_l}(\tau_l,\sigma_l)\right).$$
As $B$ is finite, the last sum converges to 0 by Theorem \ref{clt_multpdg}. As all the $2^l$ terms converge to 0 and $l$ is finite, this imples the convergence of joint cumulants of $H_{b,j}(\tau,\sigma)$ for order $l > 2$ to zero.

\end{proof}

\subsubsection*{Proof of \autoref{thm:null}}  We write $Q_{\omega_0,\delta} = \int_0^1 \int_0^1 Q_{\omega_0,\delta}(\tau,\sigma) d\tau d\sigma,$ where 
$$Q_{\omega_0,\delta}(\tau,\sigma) = \sum_{b=1}^B  \left(\widehat{g}_{b/B,\omega_0+\delta}(\tau,\sigma)- \widetilde{g}_{b/B,\omega_0,\delta}(\tau,\sigma)\right)^2.$$ Noting than the functional $:L^2\left([0,1]^2\right) \mapsto \mathbb{R}$ is continuous it is enough to show that
$$ Q_{\omega_0,\delta}(\tau,\sigma) \stackrel{d}{=} \frac{1}{K}\sum_{b=1}^B \left(\mathcal{G}_b^2(\tau,\sigma) + o_p(1)\right),$$
This will be proved in two steps. Specifically, we will show that: 
\begin{itemize}
    \item[(i)] The result holds in finite dimensional distribution, i.e., for any $k$ and any fixed $(\tau_1,\sigma_1),\dots,(\tau_k,\sigma_k) \in [0,1]^2$ the distribution of the random vector $$(Q_{\omega_0,\delta}(\tau_1,\sigma_1),\dots, Q_{\omega_0,\delta}(\tau_k,\sigma_k))$$ is asymptotically equal to the distribution of $$\frac{1}{K}\left(\sum_{b=1}^B \left(\mathcal{G}_b^2(\tau_1,\sigma_1) +o_p(1)\right), \dots, \sum_{b=1}^B \left(\mathcal{G}_b^2(\tau_1,\sigma_1) +o_p(1)\right)\right).$$
    \item[(ii)] The process $\{Q_{\omega_0,\delta}(\tau,\sigma)\}_{(\tau,\sigma) \in [0,1]^2}$ is asymptotically tight as a process in $L^2([0,1]^2).$
\end{itemize}

Without the loss of generality, we will prove (i) for $k=1$, the result for general $k$ can be proved similarly with some more notations.  Note that the scan statistics can be written as
\begin{equation*}
    Q_{\omega_0,\delta}(\tau,\sigma) = \frac{1}{K}\sum_{b=1}^B A_{b,K,T}
\end{equation*}
where
$$ A_{b,K,T}(\tau,\sigma) = K\left(\widehat{g}_{b/B,\omega_0+\delta}(\tau,\sigma)- \widetilde{g}_{b/B,\omega_0,\delta}(\tau,\sigma)\right)^2.$$
Therefore it is enough to show the process $A_{b,k,T} \in L^2([0,1]^2)$ converges in distribution to $\mathcal{G}_b^2$, where $\mathcal{G}_b$ is the Gaussian process defined in the statement of the Theorem, uniformly over $b \in \{1,2,\dots,B\}.$ In order to establish that write
\begin{align}\label{eq:A}
    A_{b,K,T}(\tau,\sigma) = &K\left(\widehat{g}_{b/B,\omega_0+\delta}(\tau,\sigma)- \E(\widehat{g}_{b/B,\omega_0+\delta}(\tau,\sigma)) - \widetilde{g}_{b/B,\omega_0,\delta}(\tau,\sigma) + \E(\widetilde{g}_{b/B,\omega_0,\delta}(\tau,\sigma))\right)^2 \nonumber\\
    &+2K\left( \widehat{g}_{b/B,\omega_0+\delta}(\tau,\sigma)  - \widetilde{g}_{b/B,\omega_0,\delta}(\tau,\sigma)\right)\left( \E(\widehat{g}_{b/B,\omega_0+\delta}(\tau,\sigma))  - \E(\widetilde{g}_{b/B,\omega_0,\delta}(\tau,\sigma))\right) \nonumber\\
    & - K\left( \E(\widehat{g}_{b/B,\omega_0+\delta}(\tau,\sigma))  - \E(\widetilde{g}_{b/B,\omega_0,\delta}(\tau,\sigma))\right)^2.
\end{align}
Suppose $L_{\delta}$ is the number of frequencies $\omega_j$ in the interval $[\omega_0,\omega_0+\delta)$. Using \autoref{lem:ghat}, under $H_0$ we have, $$\E(\widetilde{g}_{b/B,\omega_0,\delta}(\tau,\sigma)) =  \frac{1}{L_{\delta}} \sum_{\omega_j \in [\omega_0,\omega_0+\delta)} \E(\widehat{g}_{b/B,\omega_0+\delta}(\tau,\sigma)) = g_{u_b,\omega_0+\delta}(\tau,\sigma) + O(\log(T_B)/T_B) + O(1/T).$$
Therefore by \autoref{lem:ghat}, the last two terms of \eqref{eq:A} is of the order $O\left(K \times \frac{\log^2{T_B}}{T_B^2} \right) + O\left(\frac{K}{T^2}\right)$, which converges to zero under \autoref{assymp}. Note that the order of these residuals are independent of the choice of block $b$.

The first term of \eqref{eq:A} can be written as $T^2(E_{b,j_1}(\tau,\sigma), \dots, E_{b,j_{L_{\delta}+1}}(\tau,\sigma)),$ where $E_{b,j}$ is the process defined in \autoref{clt_multpdg}, the set of frequencies $\{\omega_{j_1},\dots, \omega_{j_{L_{\delta}}}\} = \{\omega : \omega \in [\omega_0,\omega_0+\delta)\}$, $\omega_{j_{L_{\delta}+1}} = \omega_0 + \delta$ and the function $T: \mathbb{R}^{L_{\delta}+1} \mapsto \mathbb{R}$ is defined as
$$T(x_1,\dots,x_{L_{\delta}+1}) = x_{L_{\delta}+1} - \frac{1}{L_{\delta}}\sum_{i=1}^{L_{\delta}} x_i.$$
Therefore an application of the Delta method along with Theorem \ref{clt_multpdg} guarantee
$$T(E_{b,j_1}(\tau,\sigma), \dots, E_{b,j_{L_{\delta}+1}}(\tau,\sigma)) \stackrel{d}{\to} \mathcal{G}_b(\tau,\sigma)$$
where $\mathcal{G}_b$ is a zero mean Gaussian process with covariance kernel given in \autoref{thm:null}. The weak convergence of $A_{b,K,T}(\tau,\sigma)$ to $\mathcal{G}_b^2(\tau,\sigma)$ then follows by the Continuous Mapping Theorem and Slutzky’s Theorem. This in turn proves the asymptotic finite dimensional distributional equivalence in (i).

Part (ii) follows from \autoref{cumfft} by (4.3) of Theorem 2 from \citep{cremers1986}. 
\subsubsection*{Proof of \autoref{thm:alt}} The proof is similar to the proof of \autoref{thm:null}. The only difference is in the treatment of the residual (second and third) terms of $A_{b,K,T}(\tau,\sigma)$ defined in \eqref{eq:A}. Note that under the alternative specified in the statement of the Theorem,
\begin{align*}
    \E(\widetilde{g}_{b/B,\omega_0,\delta}(\tau,\sigma)) =  &\frac{1}{L_{\delta}} \sum_{\omega_j \in [\omega_0,\omega_0+\delta)} \E(\widehat{g}_{b/B,\omega_0+\delta}(\tau,\sigma)) \\
    = &\frac{\#\{\omega: \omega \in (\omega_0,\omega^*)\}}{L_{\delta}}  g^{(1)}_{u_b}(\tau,\sigma) + \frac{\#\{\omega: \omega \in [\omega^*,\omega_0+\delta)\}}{L_{\delta}}  g^{(2)}_{u_b}(\tau,\sigma)\\
    &~~+ O(\log(T_B)/T_B) + O(1/T).\\
    = &\frac{\omega^* - \omega_0}{\delta}g^{(1)}_{u_b}(\tau,\sigma) + \frac{\omega_0 + \delta - \omega^*}{\delta}g^{(2)}_{u_b}(\tau,\sigma) + O(\log(T_B)/T_B) + O(1/T).\\
    = &g^{(2)}_{u_b}(\tau,\sigma) + \frac{\omega^* - \omega_0}{\delta}\left(g^{(1)}_{u_b}(\tau,\sigma) - g^{(2)}_{u_b}(\tau,\sigma)\right) + O(\log(T_B)/T_B) + O(1/T).
\end{align*}

Note that the second equality follows from the fact that $\omega_j$'s are chosen equally spaced. Therefore the residual is

$$K\left( \E(\widehat{g}_{b/B,\omega_0+\delta}(\tau,\sigma))  - \E(\widetilde{g}_{b/B,\omega_0,\delta}(\tau,\sigma))\right)^2 + o_p(1) = \frac{\omega^* - \omega_0}{\delta}\left(g^{(1)}_{u_b}(\tau,\sigma) - g^{(2)}_{u_b}(\tau,\sigma)\right) + o_p(1).$$
The the rest of the proof is similar the proof of \autoref{thm:null}.

\subsubsection*{Proof of \autoref{thm:sttest}} We write 
$$Q^0(\omega_1,\omega_2) = \frac{1}{K}\sum_{b=1}^B \int_{\omega_1}^{\omega_2} \int_0^1 \int_0^1 \left(\sqrt{K}\widehat{g}_{b/B,\omega}(\tau,\sigma)\right)^2 d\tau d\sigma d\omega = \int_{\omega_1}^{\omega_2} \int_0^1 \int_0^1 Q_0(\omega,\tau,\sigma),$$
where $$Q_0(\omega,\tau,\sigma) = \frac{1}{K}\sum_{b=1}^B \left(\sqrt{K}\widehat{g}_{b/B,\omega}(\tau,\sigma)\right)^2.$$
An application of Lemma \ref{lem:ghat} along with continuous mapping theorem guarantees under $H_0$ that the finite dimensional distributions of $Q_0(\omega,\tau,\sigma)$ are asymptotically equivalent to $\frac{1}{K}\sum_{b=1}^B  \mathcal{H}_b^2(\tau,\sigma).$ The rest of the proof follows  is similar as to the proof of Theorem \ref{thm:null}.

\subsubsection*{Proof of \autoref{alt_st}} Using similar expansion as in the proof of Theorem \ref{thm:sttest}, under the alternative we write
\begin{align*}
   Q_0(\omega,\tau,\sigma) = &\frac{1}{K}\sum_{b=1}^B \left(\sqrt{K}\widehat{g}_{b/B,\omega}(\tau,\sigma)\right)^2 \\
   = & \frac{1}{K}\sum_{b=1}^B \left(\sqrt{K}\left[\widehat{g}_{b/B,\omega}(\tau,\sigma) - g_{u}(\tau,\sigma) + g_{u}(\tau,\sigma)\right]\right)^2 \\
   = &\frac{1}{K}\sum_{b=1}^B \left(\sqrt{K}\left[ \widehat{g}_{b/B,\omega}(\tau,\sigma) - g_{u}(\tau,\sigma)\right]\right)^2  + Bg_u^2(\tau,\sigma)  \\
   &+  \frac{2 g_{u}(\tau,\sigma)}{\sqrt{K}}\sum_{b=1}^B \sqrt{K}\left[\widehat{g}_{b/B,\omega}(\tau,\sigma) - g_{u}(\tau,\sigma)\right] . 
\end{align*}
As the second term dominates under the asymptotic scheme in \autoref{assymp}, the quantity
$Q_0(\omega,\tau,\sigma) = O_p(B)$. Taking integral over $\omega, \tau, \sigma$ we have $Q^0(\omega_1,\omega_2) = O_p(B)$ and the result follows.

\subsection{Asymptotic Distribution Properties}
By  Lemma \ref{lem:ghat} the covariance of the process $\mathcal{H}_b$ defined in Theorem \ref{thm:sttest} is given by
\begin{align}\label{covH}
    \text{Cov}(\mathcal{H}_b(\tau_1,\sigma_1),\mathcal{H}_b(\tau_2,\sigma_2)) = C(b_1,b_2,\omega_1,\omega_2,\tau_1,\sigma_1,\tau_2,\sigma_2),
\end{align}
where
\begin{align}
    C(b,b,\omega_1,\omega_2,\tau_1,\sigma_1,\tau_2,\sigma_2)
    = &(1-2/B)[f_{u_b,\omega_1}(\tau_1,\tau_2)f_{u_b,\omega_2}(\sigma_1,\sigma_2)+f_{u_b,\omega_1}(\tau_1,\sigma_2)f_{u_b,\omega_2}(\tau_2,\sigma_1)] \label{C}\\
    &+\frac{1}{B^2}\sum_{l=1}^B [f_{u_l,\omega_1}(\tau_1,\tau_2)f_{u_l,\omega_2}(\sigma_1,\sigma_2)+f_{u_l,\omega_1}(\tau_1,\sigma_2)f_{u_l,\omega_2}(\tau_2,\sigma_1)]\nonumber\\
 {C}\left(b_1,b_2,\omega_1,\omega_2,\tau_1,\sigma_1,\tau_2,\sigma_2\right)
    =  &-\frac{1}{B}\left[f_{u_{b_1},\omega_1}(\tau_1,\tau_2)f_{u_{b_1},\omega_2}(\sigma_1,\sigma_2)+ f_{u_{b_1},\omega_1}(\tau_1,\sigma_2)f_{u_{b_1},\omega_2}(\tau_2,\sigma_1)\right]\nonumber\\
    &- \frac{1}{B}\left[f_{u_{b_2},\omega_1}(\tau_1,\tau_2)f_{u_{b_2},\omega_2}(\sigma_1,\sigma_2) + f_{u_{b_2},\omega_1}(\tau_1,\sigma_2)f_{u_{b_2},\omega_2}(\tau_2,\sigma_1)\right]\nonumber\\
    &~+ \frac{1}{B^2}\sum_{l=1}^B \left[f_{u_{b_l},\omega_1}(\tau_1,\tau_2)f_{u_l,\omega_2}(\sigma_1,\sigma_2) + f_{u_{b_l},\omega_1}(\tau_1,\sigma_2)f_{u_{b_l},\omega_2}(\tau_2,\sigma_1)\right]\nonumber,
\end{align}
where $b_1 \neq b_2$.\\
The covariance structure of the process $\mathcal{G}_b$ is given by
\begin{align}\label{var:diagonal}
    Cov(\mathcal{G}_b(\tau_1,\sigma_1),\mathcal{G}_b(\tau_2,\sigma_2)) = &C(b,b,\omega_0+\delta,\omega_0+\delta,\tau_1,\sigma_1,\tau_2,\sigma_2))\nonumber\\
    &+ \frac{1}{L_{\delta}^2}\sum_{j=1}^{L_{\delta}}\sum_{k=1}^{L_{\delta}}C(b,b,\omega_j,\omega_k,\tau_1,\sigma_1,\tau_2,\sigma_2))\nonumber\\
    &- \frac{1}{L_{\delta}}\sum_{j=1}^L C(b,b,\omega_0+\delta,\omega_J,\tau_1,\sigma_1,\tau_2,\sigma_2))\nonumber\\
    &- \frac{1}{L_{\delta}}\sum_{j=1}^L C(b,b,\omega_j,\omega_0+\delta,\tau_1,\sigma_1,\tau_2,\sigma_2)),
\end{align}
and for $b_1 \neq b_2$,
\begin{align}\label{var:offdiag}
    Cov(\mathcal{G}_{b_1}(\tau_1,\sigma_1),\mathcal{G}_{b_2}(\tau_2,\sigma_2)) = &C(b_1,b_2,\omega_0+\delta,\omega_0+\delta,\tau_1,\sigma_1,\tau_2,\sigma_2))\nonumber\\
    &+ \frac{1}{L_{\delta}^2}\sum_{j=1}^{L_{\delta}}\sum_{k=1}^{L_{\delta}}C(b_1,b_2,\omega_j,\omega_k,\tau_1,\sigma_1,\tau_2,\sigma_2))\nonumber\\
    &- \frac{1}{L_{\delta}}\sum_{j=1}^L C(b_1,b_2,\omega_0+\delta,\omega_j,\tau_1,\sigma_1,\tau_2,\sigma_2)))\nonumber\\
    &- \frac{1}{L_{\delta}}\sum_{j=1}^L C(b_1,b_2,\omega_j,\omega_0+\delta,\tau_1,\sigma_1,\tau_2,\sigma_2))
\end{align}
where, $C$ is as defined in \eqref{C}.

\begin{lemma}\label{lem:gb}
As $B \to \infty$, the quantities $\frac{1}{B}\sum_{b=1}^B\|\mathcal{G}_b\|^2$ and $\frac{1}{B}\sum_{b=1}^B\|\mathcal{H}_b\|^2$ are $O_p(1).$
\end{lemma}

\begin{proof} We will show this for the case of $\mathcal{G}_b$. The proof for $\mathcal{H}_b$ is similar.

Note that
\begin{align*}
    &\E\left(\frac{1}{B}\sum_{b=1}^B\|\mathcal{G}_b\|^2\right) = \frac{1}{B}\sum_{b=1}^B\E\left(\|\mathcal{G}_b\|^2\right)\\
    &= \frac{1}{B}\sum_{b=1}^B\E\left(\int_{0}^1\int_0^1 \mathcal{G}_b^2(\tau,\sigma) d\tau d\sigma \right)\\
    & =  \frac{1}{B}\sum_{b=1}^B\int_{0}^1\int_0^1 \E \mathcal{G}_b^2(\tau,\sigma) d\tau d\sigma \\
    & \to \int_{0}^1\int_0^1\int_{0}^1 \left[f_{u,\omega_{0}+\delta}(\tau,\tau)f_{u,\omega_{0}+\delta}(\sigma,\sigma) + f_{u,\omega_{0}+\delta}^2(\tau,\sigma)\right]d\tau d\sigma du\\
    &~~+\frac{1}{L_{\delta}^2}\sum_{j=1}^{L_{\delta}}\sum_{k=1}^{L_{\delta}}\iiint \left[f_{u,\omega_j}(\tau,\tau)f_{u,\omega_k}(\sigma,\sigma) + f_{u,\omega_j}(\tau,\sigma)f_{u,\omega_k}(\tau,\sigma)\right]d\tau d\sigma du\\
    &~~-\frac{2}{L_{\delta}}\sum_{j=1}^{L_{\delta}}\iiint \left[f_{u,\omega_j}(\tau,\tau)f_{u,\omega_0+\delta}(\sigma,\sigma) + f_{u,\omega_j}(\tau,\sigma)f_{u,\omega_0+\delta}(\tau,\sigma)\right]d\tau d\sigma du,
\end{align*}
as $B \to \infty$. As $f_{u,\omega}(.)$ is continuous in $u$ and $f_{u,\omega}$ is square integrable for all $u \in [0,1]$ and $\omega \in (0,0.5]$, the integrals in the limit are finite. Note that we can exchange the integral and expectation in the second line by Fubini's Theorem as the double integral is finite. 
\begin{align*}
    Var\left(\frac{1}{B}\sum_{b=1}^B\|\mathcal{G}_b\|^2\right) = \frac{1}{B^2} \sum_{b=1}^B Var\left(\|\mathcal{G}_b\|^2\right) + \frac{1}{B^2} \sum_{b_1=1}^B\sum_{b_2=1}^B Cov\left(\|\mathcal{G}_{b_1}\|^2,\|\mathcal{G}_{b_2}\|^2\right)
\end{align*}
The first term can be simplified as
\begin{align*}
    \frac{1}{B^2} \sum_{b=1}^B Var\left(\|\mathcal{G}_b\|^2\right) &\leq \frac{1}{B^2} \sum_{b=1}^B \E\left(\|\mathcal{G}_b\|^4\right) = \frac{1}{B^2} \sum_{b=1}^B \E \left(\int_0^1\int_0^1 \mathcal{G}_b^2(\tau,\sigma)d\tau d\sigma\right)^2\\
    & \leq \frac{1}{B^2} \sum_{b=1}^B  \left(\int_0^1\int_0^1 \E \mathcal{G}_b^2(\tau,\sigma)d\tau d\sigma\right)^2 = O(1/B).
\end{align*}
Similarly with some standard algebra we can show that
\begin{align*}
    \frac{1}{B^2} \sum_{b_1=1}^B\sum_{b_2=1}^B Cov\left(\|\mathcal{G}_{b_1}\|^2,\|\mathcal{G}_{b_2}\|^2\right) = \frac{1}{B^2} \sum_{b_1=1}^B\sum_{b_2=1}^B \E\left(\|\mathcal{G}_{b_1}\|^2\|\mathcal{G}_{b_2}\|^2\right) = O(1/B).
\end{align*}
Therefore as $B \to \infty$, $\E\left(\frac{1}{B}\sum_{b=1}^B\|\mathcal{G}_b\|^2\right) = O(1)$ and $Var\left(\frac{1}{B}\sum_{b=1}^B\|\mathcal{G}_b\|^2\right) \to 0$ and hence $\frac{1}{B}\sum_{b=1}^B\|\mathcal{G}_b\|^2$ is $O_p(1)$.
\end{proof}
\end{appendix}

%%  MR numbers will be added by VTeX.                      %%
%%                                                         %%
%%  Use \cite{...} to cite references in text.             %%
%%                                                         %%
%%%%%%%%%%%%%%%%%%%%%%%%%%%%%%%%%%%%%%%%%%%%%%%%%%%%%%%%%%%%%

%% or include bibliography directly:
% \begin{thebibliography}{}
% \bibitem{b1}
% \end{thebibliography}

%%%%%%%%%%%%%%%%%%%%%%%%%%%%%%%%%%%%%%%%%%%%%%
%% Support information (funding), if any,   %%
%% should be provided in the                %%
%% Acknowledgements section.                %%
%%%%%%%%%%%%%%%%%%%%%%%%%%%%%%%%%%%%%%%%%%%%%%
\section*{Acknowledgements}
Research
reported in this publication was supported by the National Institute Of General Medical Sciences
of the National Institutes of Health under Award Number R01GM140476. The content is solely the
responsibility of the authors and does not necessarily represent the official views of the National
Institutes of Health.
 
%%%%%%%%%%%%%%%%%%%%%%%%%%%%%%%%%%%%%%%%%%%%%%
%% Supplementary Material, if any, should   %%
%% be provided in {supplement} environment  %%
%% with title inside \textbf{} and short    %%
%% description below.                       %%
%%%%%%%%%%%%%%%%%%%%%%%%%%%%%%%%%%%%%%%%%%%%%%
\begin{supplement}[id=suppA]
\stitle{\textbf{\texttt{R} code} for "Adaptive Frequency Band Analysis for Functional Time Series"}
\slink[url]{https://github.com/sbruce23/fEBA}
\sdescription{\texttt{R} code, a quick start demo, and descriptions of all functions and parameters needed to generate simulated data introduced in Section \ref{sec:sims} and to implement the proposed method on data for use in practice can be downloaded from GitHub at this link.}
\end{supplement}

%%%%%%%%%%%%%%%%%%%%%%%%%%%%%%%%%%%%%

%% if your bibliography is in bibtex format, uncomment commands:
% \bibliographystyle{imsart-number} % Style BST file (imsart-number.bst or imsart-nameyear.bst)
\bibliography{functionalfreqband}       % Bibliography file (usually '*.bib')

\end{document}